\documentclass[11pt,a4paper]{article}
\usepackage[top=70pt,bottom=70pt,left=60pt,right=60pt]{geometry}
\usepackage{amsmath,amsfonts,amssymb,amsthm}
\usepackage{mathtools}
\usepackage{bm}
\usepackage{braket}
\usepackage{tikz}
\usepackage{subfigure}
\usepackage{graphicx,caption}

\usepackage{tikz}
\usepackage{pgfplots}

\usepackage[unicode=true,
bookmarks=true,bookmarksopen=false,
breaklinks=false,pdfborder={0 0 0},colorlinks=true]
{hyperref}

\usepackage{xcolor}
\definecolor{cblue}{rgb}{0.16, 0.32, 0.75}
\definecolor{cred}{rgb}{0.7, 0.11, 0.11}
\hypersetup{%
	,linkcolor=cred
	,citecolor=cblue
}

\usepackage{authblk}

\usepackage{placeins}

\usetikzlibrary{decorations.markings,shapes.misc,intersections, pgfplots.fillbetween}
\tikzset{->-/.style={decoration={
			markings,
			mark=at position .5 with {\arrow{<}}},postaction={decorate}}}
\tikzset{-<-/.style={decoration={
			markings,
			mark=at position .5 with {\arrow{>}}},postaction={decorate}}}
\pgfplotsset{compat=1.16}
 
\theoremstyle{plain}
\newtheorem{theorem}{Theorem}[section]

\newtheorem{lemma}[theorem]{Lemma}
\newtheorem{proposition}[theorem]{Proposition}

\newtheorem{hypothesis}{Hypothesis}

\theoremstyle{definition}

\theoremstyle{remark}

\newtheorem{example}[theorem]{Example}

\renewcommand{\i}{\mathrm{i}}
\newcommand{\e}{\mathrm{e}}
\renewcommand{\Re}{\operatorname{Re}}
\renewcommand{\Im}{\operatorname{Im}}

\newcommand{\R}{\mathbb{R}}
\renewcommand{\C}{\mathbb{C}}
\renewcommand{\S}{\mathsf{\Sigma}}
\renewcommand{\G}{\mathsf{G}}
\newcommand{\K}{\mathcal{K}}
\newcommand{\A}{\mathcal{A}}
\newcommand{\Z}{\mathsf{Z}}

\renewcommand{\Re}[0]{{\mathrm{Re}\,}}
\renewcommand{\Im}[0]{{\mathrm{Im}\,}}

\usepackage{etoolbox}
\makeatletter
\def\@mkboth#1#2{}
\newlength\appendixwidth
\preto\appendix{\addtocontents{toc}{\protect\patchl@section}}
\newcommand{\patchl@section}{%
	\settowidth{\appendixwidth}{\textbf{Appendix }}%
	\addtolength{\appendixwidth}{1.5em}%
	\patchcmd{\l@section}{1.5em}{\appendixwidth}{}{\ddt}%
}
\makeatother

\begin{document}

\title{\textbf{The self-energy of Friedrichs-Lee models\\and its application to bound states and resonances}}

\author[$1,2,\star$]{Davide Lonigro}

\affil[$1$]{\small Dipartimento di Fisica and MECENAS, Universit\`a di Bari, I-70126 Bari, Italy}
\affil[$2$]{\small INFN, Sezione di Bari, I-70126 Bari, Italy}
\affil[$\star$]{\small\texttt{davide.lonigro@ba.infn.it}}

\date{}
\maketitle

\begin{abstract}
A system composed of two-level systems interacting with a single excitation of a one-dimensional boson field with continuous spectrum, described by a Friedrichs (or Friedrichs-Lee) model, can exhibit bound states and resonances; the latter can be characterized by computing the so-called self-energy of the model. We evaluate an analytic expression, valid for a large class of dispersion relations and coupling functions, for the self-energy of such models. Afterwards, we focus on the case of identical two-level systems, and we refine our analysis by distinguishing between dominant and suppressed contributions to the associated self-energy; we finally examine the phenomenology of bound states in the presence of a single dominant contribution.
\end{abstract}

\section{Introduction}

The study of the interaction between a structured boson field and a family of two-level systems, or atoms, is an ubiquitous topic in both mathematical and theoretical physics \cite{ingold,breuer,leggett}; the dependence of the spectral properties of the system and, in particular, the insurgence of bound states and resonances and their dependence on the structure of the coupling, may be of primary relevance. 

A simplified, yet largely useful, description of this scenario is given by the Lee model \cite{Lee,Giacosa,Giacosa2}, which was originally introduced as an example of a renormalizable solvable quantum-field model; in the Lee model, the interaction between the atoms and the boson is constructed in such a way that the global excitation number is conserved, i.e., physically, a boson is created in the field if and only if one of the atoms switches from its excited state to the ground one, and vice versa. When the quantum field is assumed to be monochromatic, the Lee model reduces to the well-known Jaynes-Cummings model \cite{Jaynes}, which is of paramount importance in quantum optics and in the description of trapped ions.

More generally, in quantum electrodynamics, a Lee structure emerges when performing a ``rotating-wave approximation": terms that do not conserve the excitation number are neglected. \cite{cohen} The model has been successfully been applied in waveguide quantum electrodynamics to study bound and scattering states that emerge when quantum emitters are coupled with a photon field confined in a quasi-1d waveguide. \cite{bic1,bic2,palma,cicc}

Finally, the presence of a conserved excitation number implies that the Fock space of the theory is naturally decomposed into excitation sectors which evolve independently, and the restriction of the Lee Hamiltonian on each sector is a well-defined self-adjoint operator whose spectral properties can be investigated; when restricting the Lee model to the single-excitation sector, we obtain an operator denoted as a Friedrichs, Lee-Friedrichs, or Friedrichs-Lee model \cite{Friedrichs,Horwitz}. At the mathematical level, such a model can be loosely described as an operator $H=H_0+V$, in which $H_0$ is a free Hamiltonian and $V$ is a perturbation which couples two different parts of its spectrum, typically a discrete and a continuous one; at the level of dynamics, this coupling brings about instability of originally stable states, which can even decay into a common final state, thus yielding mixing between the original ones.

 The model has found applications in different frameworks, such as quantum field theory~\cite{Araki}, non-relativistic quantum electrodynamics~\cite{Frohlich}, quantum optics~\cite{Gadella}, quantum probability~\cite{vonWaldenfels}, and hadron resonance phenomenology \cite{liu,zhou}, to name a few, and has proven to be an useful testbed for fundamental phenomena like the Zeno and anti-Zeno effects \cite{zeno,antizeno}. An extended description of the spectral properties of the single-atom Friedrichs-Lee model was provided in \cite{FLProc,FL}.

As it turns out, the spectral properties of Friedrichs-Lee models (and, in particular, their bound states and resonances) with $n$ atoms are entirely determined by an $n\times n$ matrix-valued function, which will be referred to as the self-energy matrix, strictly dependent on the choice of the coupling between the atoms and the field as well as the structure of the field. When describing the system in the momentum representation, the self-energy is expressed as an integral over the momentum space. In principle, if the self-energy is explicitly known, all bound states and resonances can be evaluated at least numerically. 

This article aims at addressing the general problem of evaluating the self-energy of Friedrichs-Lee models under some regularity assumptions on the structure of the field and the coupling. Namely, we will restrict our attention to the case in which the boson field is one-dimensional, has a continuous spectrum, and its momentum space coincides with the full real line, the concrete example of such a situation being a photon field confined in an infinite linear waveguide. 

We stress that, in our calculations, no specific expression for neither the dispersion relation nor the coupling functions modeling the atom-field interaction shall be assumed, modulo some technical requirements; an effort in this direction is justified, aside from the theoretical interest \textit{per se} of the model, by the enormous developments of quantum technologies in recent years, which pave the way toward engineering quantum systems characterized by diverse values of both parameters.

This article is organized as follows:
\begin{itemize}
	\item in Section~\ref{sec:general} we describe the structure of an $n$-atom Friedrichs-Lee model: we set up the Schr\"odinger equation for the model, introduce the self-energy, and briefly discuss the emergence of bound states and resonances;
	\item in Section~\ref{sec:self} we evaluate an explicit expression for the self-energy (Prop.~\ref{prop:prop}) under some assumptions on the parameters of the model, thoroughly discussed in the section;
	\item in Section~\ref{sec:identical} we particularize our discussion to a class of coupling functions which can be physically associated to identical atoms in an one-dimensional infinite geometry; the bound states emerging in a subclass of such models are described. 
\end{itemize}

\section{Friedrichs-Lee models and their properties}\label{sec:general}
We will start our analysis by recalling the definition of Friedrichs-Lee models, and by discussing their spectral properties; we will show that the latter depend crucially on a matrix-valued function, the self-energy, whose calculation will be the core of the following sections.

\subsection{Preliminaries}
Let $\omega:\mathbb{R}\rightarrow\mathbb{R}$ a continuous real-valued function bounded from below, i.e.  $\min_{k\in\mathbb{R}}\omega(k)>-\infty$; without loss of generality, we will assume $\omega(k)\geq0$. We shall consider a free Hamiltonian whose spectrum has a simple, purely continuous component, coinciding with the range of $\omega$, plus $n$ eigenvalues $\varepsilon_1,\dots,\varepsilon_n$, to be counted with their multiplicity; formally, we can write such an operator via the following expression: 
\begin{equation}\label{eq:h0}
	H_0=\sum_{\ell=1}^{n}\varepsilon_\ell\ket{\ell}\!\bra{\ell}+\int_{\R}\omega(k)\ket{k}\!\bra{k}\,\mathrm{d}k,
\end{equation}
with the following constraints holding for all $j,\ell=1,\dots,n$ and $k,k'\in\mathbb{R}$:
\begin{equation}\label{eq:on}
	\braket{j|\ell}=\delta_{j\ell},\quad\braket{k|k'}=\delta(k-k'),\quad\braket{j|k}=0.
\end{equation}
A rigorous implementation of the formal equations above can be obtained by means of rigged Hilbert spaces \cite{Madrid,Bohm,BohmGadella,MadridBohm}; see \cite{GadellaGomez} for a review of the mathematical foundations of this notation, which also includes alternative approaches. Physically, as discussed more extensively in the introduction, here $\varepsilon_1,\dots,\varepsilon_n$ represent the excitation energies of the atoms, and $\omega(k)$ the dispersion relation of the boson field.

An interaction between the two parts of the spectrum of $H_0$ (physically, between the atoms and the field) will be implemented via the following operator:\footnote{Note that, here and in the following, $\overline{z}$ is the complex conjugate of the complex number $z$.}
\begin{equation}\label{eq:v}
	V=\sum_{\ell=1}^n\int_{\R}\Bigl(\overline{F_\ell(k)}\ket{k}\!\bra{\ell}+F_\ell(k)\ket{\ell}\!\bra{k}\Bigr)\,\mathrm{d}k,
\end{equation}
with $\{F_\ell\}_{\ell=1,\dots,n}$ being a set of continuous functions which modulate the coupling between the $\ell$th eigenvalue (physically, the $\ell$th atom) and the boson field. We will call them the \textit{form factors} of the model. For future purposes, we will require the following constraint to hold:
\begin{equation}\label{eq:normalization}
	\int_{\R}\frac{|F_j(k)|^2}{\omega(k)+1}\,\mathrm{d}k<\infty,\qquad j=1,\dots,n.
\end{equation}
Since the integrand has no singularities for finite $k$ (due to the request $\omega(k)\geq0$), Eq.~\eqref{eq:normalization} means that $|F_j(k)|^2$ cannot ``grow too fast'' with respect to the dispersion relation as $|k|\to\infty$.

\subsection{Schr\"odinger equation for Friedrichs-Lee models}
By Eq.~\eqref{eq:on}, the most generic state $\ket{\Psi}$ of the model can be expanded as follows:
\begin{equation}\label{eq:most}
\ket{\Psi}=\sum_{\ell=1}^na_\ell\ket{\ell}+\int_{\R}\xi(k)\ket{k}\:\mathrm{d}k,
\end{equation}
with $a_1,\dots,a_n\in\mathbb{C}$, and $\xi(k)$ being a square-integrable function,
\begin{equation}
	\int_{\mathbb{R}}|\xi(k)|^2\;\mathrm{d}k<\infty,
\end{equation}
with the normalization constraint
\begin{equation}
	\sum_{\ell=1}^n|a_\ell|^2+\int_{\mathbb{R}}|\xi(k)|^2\;\mathrm{d}k=1.
\end{equation}
The Schr\"odinger equation for a state as in Eq.~\eqref{eq:most}, namely,
\begin{equation}
\i\frac{\mathrm{d}}{\mathrm{d}t}\ket{\Psi(t)}=H\ket{\Psi(t)},
\end{equation}
is therefore equivalent to the coupled differential system
\begin{equation}
\begin{dcases}
\i\,\dot{a}_j(t)=\varepsilon_j\,a_j(t)+\int_{\R}\xi(k,t)F_j(k)\,\mathrm{d}k;\\
\i\,\dot{\xi}(k,t)=\omega(k)\xi(k,t)+\sum_{\ell=1}^n \overline{F_\ell(k)} a_\ell(t),
\end{dcases}
\end{equation}
the first equation holding for all $j=1,\dots,n$. It will be convenient to rewrite the system above as follows:\footnote{Hereafter, the adjoint of a column vector $\bm{a}$, i.e. the row vector whose entries are the complex conjugates of those of $\bm{a}$, will be written as $\bm{a}^\dag$; for two $n$-component column vectors $\bm{a},\bm{b}$, the scalar $\bm{a}^\dag\bm{b}$ and the $n\times n$ matrix $\bm{a}\bm{b}^\dag$ are respectively their inner and outer product.}
\begin{equation}\label{eq:convenient}
\begin{dcases}
\i\,\dot{\bm{a}}(t)=\mathcal{E}\,\bm{a}(t)+\int_{\R}\xi(k,t)\bm{F}(k)\,\mathrm{d}k;\\
\i\,\dot{\xi}(k,t)=\omega(k)\xi(k,t)+\bm{F}(k)^\dag\bm{a}(t),
\end{dcases}
\end{equation}
where we are introducing the $n\times n$ diagonal matrix $\mathcal{E}=\mathrm{diag}\left(\varepsilon_1,\dots,\varepsilon_n\right)$ and the column vectors 
\begin{eqnarray}
\bm{a}(t)=\begin{pmatrix}
a_1(t)\\a_2(t)\\\vdots\\ a_n(t)
\end{pmatrix},\qquad\bm{F}(k)=\begin{pmatrix}
F_1(k)\\F_2(k)\\\vdots\\ F_n(k)
\end{pmatrix}.
\end{eqnarray}
This system can be conveniently solved in the complex energy domain, that is, by taking the \textit{Fourier-Laplace transform} of both equations with respect to $t$:
\begin{eqnarray}\label{eq:fl1}
\hat{\bm{a}}(z)&=&\i\int_0^\infty \bm{a}(t)\,\e^{\i zt}\,\mathrm{d}t,\\\label{eq:fl2}
\hat{\xi}(k,z)&=&\i\int_0^\infty \xi(k,t)\,\e^{\i zt}\,\mathrm{d}t,
\end{eqnarray}
with $z$ being a complex variable, $\Im z>0$, that can be interpreted as a complex energy. The Fourier-Laplace transformed system reads
\begin{equation}
\begin{dcases}
z\,\hat{\bm{a}}(z)+\bm{a}_0=\mathcal{E}\,\hat{\bm{a}}(z)+\int_{\R}\hat{\xi}(k,z)\bm{F}(k)\;\mathrm{d}k;\\
z\,\hat{\xi}(k,z)+\xi_0(k)=\omega(k)\hat{\xi}(k,z)+\bm{F}(k)^\dag\hat{\bm{a}}(z),
\end{dcases}
\end{equation}
where we have set $\bm{a}(0)\equiv\bm{a}_0$, $\xi(k,0)\equiv\xi_0(k)$. The second equation yields
\begin{equation}\label{eq:chiz}
\hat{\xi}(k,z)=\frac{\xi_0(k)}{\omega(k)-z}-\frac{\bm{F}(k)^\dag\hat{\bm{a}}(z)}{\omega(k)-z};
\end{equation}
by substituting into the first one, we get
\begin{equation}\label{eq:eqaz}
\bigl(\mathcal{E}-z\,\mathsf{Id}-\mathsf{\Sigma}(z)\bigr)\hat{\bm{a}}(z)=\bm{a}_0-\int_{\R}\frac{\xi_0(k)}{\omega(k)-z}\bm{F}(k)\,\mathrm{d}k,
\end{equation}
where $\mathsf{Id}$ is the $n$-dimensional identity matrix, and we define the \textit{self-energy matrix} via
\begin{equation}
\mathsf{\Sigma}(z)=\int_{\R}\frac{\bm{F}(k)\bm{F}(k)^\dagger}{\omega(k)-z}\,\mathrm{d}k,
\end{equation}
that is, $\left[\mathsf{\Sigma(z)}\right]_{j\ell}=\Sigma_{j\ell}(z)$, with
\begin{equation}
\Sigma_{j\ell}(z)=\int_{\R}\frac{F_j(k)\overline{F_\ell(k)}}{\omega(k)-z}\,\mathrm{d}k,
\end{equation}
this matrix being well-defined, as it may be readily shown, if the constraint in Eq.~\eqref{eq:normalization} holds.

A simple calculation shows that, whenever $\Im z>0$, the anti-hermitian part $\Im\Sigma(z)$ of the self-energy is a positive semidefinite matrix, i.e. it is a matrix-valued Nevanlinna-Herglotz function (see e.g. \cite{Gesztesy2} for a survey of the properties of such objects); consequently, the anti-hermitian part of the matrix $-\mathcal{E}+z\mathsf{Id}+\S(z)$ is positive definite, and thus has a strictly nonzero determinant. Since a matrix with positive definite (anti--)hermitian part is nonsingular \cite{Mathias}, Eq.~\eqref{eq:eqaz} admits a unique solution given by
\begin{equation}\label{eq:az}
\hat{\bm{a}}(z)=\bigl[\mathcal{E}-z\,\mathsf{Id}-\mathsf{\Sigma}(z)\bigr]^{-1}\left[\bm{a}_0-\int_{\R}\frac{\xi_0(k)}{\omega(k)-z}\bm{F}(k)\,\mathrm{d}k\right].
\end{equation}
Substituting Eq.~\eqref{eq:az} into~\eqref{eq:chiz} finally yields $\hat{\xi}(k,z)$ and thus the full solution of the Schr\"odinger equation in the complex energy domain. The solution in the time domain can be finally reconstructed by reversing Eqs.~\eqref{eq:fl1}--\eqref{eq:fl2}, that is,
\begin{eqnarray}\label{eq:sol}
\bm{a}(t)&=&\frac{1}{2\pi\i}\int_{\i\delta-\infty}^{\i\delta+\infty}\hat{\bm{a}}(z)\,\e^{-\i zt}\,\mathrm{d}z;\\\label{eq:sol2}
\xi(k,t)&=&\frac{1}{2\pi\i}\int_{\i\delta-\infty}^{\i\delta+\infty}\hat{\xi}(k,z)\,\e^{-\i zt}\,\mathrm{d}z,
\end{eqnarray}
with $\delta>0$ being arbitrary. Many quantities of physical interest can be computed from these expressions: in particular, the survival amplitude of any initial state characterized by $\xi_0(k)\equiv0$ (i.e. a state in which the excitation is entirely shared by the atoms), corresponding to $\mathcal{A}(t)=\bm{a}_0^\dagger\bm{a}(t)$, reads
\begin{equation}\label{eq:at}
\mathcal{A}(t)=\frac{1}{2\pi\i}\int_{\i\delta-\infty}^{\i\delta+\infty}\bm{a}_0^\dag\,\bigl[\mathcal{E}-z\,\mathsf{Id}-\mathsf{\Sigma}(z)\bigr]^{-1}\bm{a}_0\,\e^{-\i zt}\;\mathrm{d}z.
\end{equation}
This discussion shows the primary role of the self-energy matrix in determining the properties of a Friedrichs-Lee model: in principle, by evaluating $\mathsf{\Sigma}(z)$, we can reconstruct the evolution of an arbitrary state of the model. 

\subsection{Bound states and resonances}
Let us have a closer look at Eq.~\eqref{eq:sol}: it consists of the integral of a function in the upper complex half-plane $\mathbb{C}^+$, given by 
\begin{equation}
	z\in\mathbb{C}^+\mapsto \bm{a}_0^\dag\,\bigl[\mathcal{E}-z\,\mathsf{Id}-\mathsf{\Sigma}(z)\bigr]^{-1}\bm{a}_0\,\e^{-\i zt}\in\mathbb{C},
\end{equation}
on a straight path above the real line. The latter function, for every choice of the initial atom configuration $\bm{a}_0$, is analytic since $z\mapsto\S(z)$ is itself (entrywise) analytic and, as discussed above, the matrix $-\mathcal{E}+z\,\mathsf{Id}+\S(z)$ is nonsingular. We may be tempted to evaluate the integral in Eq.~\eqref{eq:sol} by closing the integration contour on a semicircle in $\mathbb{C}^+$, but the presence of the term $\e^{-\i zt}$, whose modulus grows exponentially as $|z|\to\infty$, prohibits such an approach.

However, let us suppose that the self-energy admits a \textit{meromorphic continuation} $\S^{(\mathrm{II})}(z)$ from $\mathbb{C}^+$ to a subset $S\subset\mathbb{C}^-$ of the lower half-plane $\mathbb{C}^-$ connected with the upper one. In this case, the function
\begin{equation}
z\in\mathbb{C}^+\cup S\mapsto \bm{a}_0^\dag\,\bigl[\mathcal{E}-z\,\mathsf{Id}-\mathsf{\Sigma}^{(\mathrm{II})}(z)\bigr]^{-1}\bm{a}_0\,\e^{-\i zt}\in\mathbb{C}
\end{equation}
is itself a meromorphic continuation of the function above and, in particular, may have \textit{poles} in $S$, given by the solutions of the equation
\begin{equation}\label{eq:resonances}
\det\left(\mathcal{E}-z\,\mathsf{Id}-\S^{(\mathrm{II})}(z)\right)=0.
\end{equation}
If Eq.~\eqref{eq:resonances} admits a family of solutions $\hat{z}_1,\dots,\hat{z}_r$ in $S$, then, via the residue theorem, the function $\mathcal{A}(t)$ can be decomposed as such:
\begin{equation}
	\mathcal{A}(t)=\sum_{s=1}^r W_s\,\e^{-\i\hat{z}_st}+\mathcal{A}_{\mathrm{contour}}(t)
\end{equation}
for some $W_1,\dots,W_r\in\mathbb{C}$, with $\mathcal{A}_{\mathrm{contour}}(t)$ being a contour contribution. Clearly, each pole $\hat{z}_s$ is associated with an exponential contribution to $\mathcal{A}(t)$ whose half-life coincides with $\Im\,\hat{z}_s$, corresponding to a \textit{resonance} of the model. 

In particular, \textit{real-valued} resonances yield oscillatory contributions to $\mathcal{A}(t)$, and thus correspond to bound states. They are characterized by
\begin{equation}\label{eq:poles}
\det\bigl(\mathcal{E}-E\,\mathsf{Id}-\S(E+\i0)\bigr)=0,
\end{equation}
with $\S(E+\i0)=\lim_{\delta\downarrow0}\S(E+\i\delta)$; the latter equation may be equivalently found, in a more direct way, by solving the eigenvalue equation for the Friedrichs-Lee model,
\begin{equation}
	H\ket{\Psi}=E\ket{\Psi},\qquad E\in\mathbb{R},
\end{equation}
with a similar strategy as the one applied in the previous paragraph for the Schr\"odinger equation. In particular, when expanding $\ket{\Psi}$ as in Eq.~\eqref{eq:most}, the column vector $\bm{a}$ corresponding to a bound state of the model with energy $E$ (physically, the excitation profile of the atoms) must satisfy
\begin{equation}
	\bigl(\mathcal{E}-E\,\mathsf{Id}-\S(E+\i0)\bigr)\bm{a}=\bm{0},
\end{equation}
the latter equation admitting nonzero solutions because of Eq.~\eqref{eq:poles}.

The resulting picture can be summarized as follows. In the absence of coupling, the model admits $n$ eigenvalues $\varepsilon_1,\dots,\varepsilon_n$ (to be counted with their multiplicity), each of them corresponding to a bound state. When switching on the coupling, those eigenvalues will generally convert into complex resonances (pictorially, they will ``sink into the complex plane"), corresponding to unstable states of the system: these are the solutions of Eq.~\eqref{eq:resonances}. Still, it \textit{may} happen that, for specific values of the parameters of the model, some of these resonances are real (they ``emerge"), thus corresponding again to genuinely bound states; precisely, bound states in the continuum (BIC). This happens at the solutions of Eq.~\eqref{eq:poles}, which is thus a particular case of Eq.~\eqref{eq:resonances}: the latter yields \textit{all} the possible resonances of the system, while the former \textit{only} yields real-valued resonances, if any.

It is therefore clear that evaluating the self-energy $\mathsf{\Sigma}(z)$ is crucial to determine whether the system admits bound states and resonances, what is their energy, and what are their features. The following section will be devoted to this goal.

\section{An explicit expression for the self-energy}\label{sec:self}

We shall now evaluate an expression for the self-energy $\mathsf{\Sigma}(z)$ of a Friedrichs-Lee model, and in particular its boundary values $\mathsf{\Sigma}(E+\i0)$, under some minimal technical assumptions; we will be able to decompose $\mathsf{\Sigma}(z)$ as the sum of \textit{pole} and \textit{contour} terms. This analysis will be refined in Section~\ref{sec:identical}, where we will specialize ourselves to a particular class of Friedrichs-Lee models describing identical atoms coupled with a boson field. 

\subsection{Well-definiteness and symmetry hypotheses}
Let us first recall the basic assumptions on $\omega(k)$ and $\bm{F}(k)$ from which we started.
\begin{hypothesis}\label{hyp1}
$\omega(k)$ is a continuous real-valued function, with $\min_{k\in\mathbb{R}}\omega(k)\geq0$, and Eq.~\ref{eq:normalization} holds.
\end{hypothesis}
Notice that, while we are merely requiring continuity at this stage, we will ultimately tighten our regularity assumptions on it. Secondly, we will make some assumptions about the \textit{symmetry properties} of the parameters.
\begin{hypothesis}\label{hyp2}
	\begin{equation}
	\omega(-k)=\omega(k)\quad\text{and}\quad F_j(-k)=\overline{F_j(k)},\;j=1,\dots,n.
	\label{eq:sim2}
	\end{equation}
\end{hypothesis}
Physically, Eq.~\eqref{eq:sim2} encodes the invariance of the model under reflections of the momentum. While not strictly required for our approach to be applicable, this assumption will greatly simplify our exposition; besides, they are indeed satisfied by the class of models that we will analyze in Section~\ref{sec:identical}, as shown in the example below.
\begin{example}\label{ex1}
	Let $x_1,x_2,\dots,x_n\in\R$, and consider a family of form factors defined as such:
	\begin{equation}
	F_j:\R\mapsto\C,\qquad F_j(k)=F(k)\,\e^{\i k x_j},
	\end{equation}
	with $F(k)$ being some function such that $F(-k)=F(k)$; as we will see in Section~\ref{sec:identical}, this choice corresponds to the case in which the atoms are identical. Clearly, $F_j(-k)=\overline{F_j(k)}$.
\end{example}
From now on, it will be useful to rewrite the self-energy as
\begin{equation}\label{eq:self}
	\mathsf{\Sigma}(z)=\int_{-\infty}^{\infty}\frac{\mathsf{G}(k)}{\omega(k)-z}\,\mathrm{d}k,
\end{equation}
where we have introduced the matrix-valued function defined by
\begin{equation}\label{eq:g}
	\mathsf{G}(k)=\bm{F}(k)\bm{F}(k)^\dag,\quad G_{j\ell}(k)=F_j(k)\overline{F_\ell(k)},
\end{equation}
which, by definition, satisfies the following properties:
\begin{equation}\label{eq:hypg}
\mathsf{G}(k)^\dag=\mathsf{G}(k),\quad\mathsf{G}(-k)=\mathsf{G}(k)^\intercal,
\end{equation}
with $\G(k)^\intercal$ being its transpose: the first property holds as an immediate consequence of its definition (Eq.~\eqref{eq:g}), while the second one follows from Hypothesis~\ref{hyp2}.

Under Hypotheses~\ref{hyp1}--\ref{hyp2}, the self-energy matrix~\eqref{eq:self} satisfies the following properties:\footnote{$\mathsf{A}^\dag$ and $\mathsf{A}^\intercal$ are the adjoint and the transpose of the $n\times n$ matrix $\mathsf{A}$.}
\begin{equation}
\S(z)^\dag=\S(\overline{z})
\label{eq:sim3}
\end{equation}
and 
\begin{equation}
\S(z)^\intercal=\S(z).
\label{eq:sim4}
\end{equation}
Eq.~\eqref{eq:sim3} is an immediate consequence of $\overline{\omega(k)}=\omega(k)$ and the first property in Eq.~\eqref{eq:hypg}:
\begin{equation}
\S(z)^\dag=\int_{-\infty}^\infty\frac{\G(k)^\dag}{\overline{\omega(k)}-\overline{z}}\,\mathrm{d}k=\int_{-\infty}^\infty\frac{\G(z)}{\omega(k)-\overline{z}}\,\mathrm{d}k=\S(\overline{z}),
\end{equation}
while Eq.~\eqref{eq:sim4} follows from $\omega(-k)=\omega(k)$ and the second property in Eq.~\eqref{eq:hypg}:
\begin{eqnarray}
\S(z)^\intercal&=&\int_{-\infty}^\infty\frac{\G(k)^\intercal}{\omega(k)-z}\,\mathrm{d}k=\int_{-\infty}^\infty\frac{\G(-k)}{\omega(k)-z}\,\mathrm{d}k\nonumber\\&=&\int_{-\infty}^\infty\frac{\G(-k)}{\omega(-k)-z}\,\mathrm{d}k=\int_{-\infty}^\infty\frac{\G(k)}{\omega(k)-z}\,\mathrm{d}k\nonumber=\S(z).
\end{eqnarray}
In particular, as a consequence of Eq.~\eqref{eq:sim4}, we only need to compute $\Sigma_{j\ell}(z)$ with (for example) $j\geq\ell$.

\subsection{Analytic continuation of the parameters}
Our approach to compute the self-energy will crucially involve the transformation of integrals over the real variable $k$ into contour integrals for functions of a complex variable $\kappa$. For this reason, we will now require the parameters to admit analytic extensions in a subset of the complex plane:
\begin{hypothesis}\label{hyp3}
	There exists an open connected subset $\K\subset\C$ of the complex plane, with $\R\subset\K$ and having smooth boundary $\Gamma=\delta \K$, and two extensions in $\overline{\K}$ of the functions $\omega(k)$ and $\G(k)$:
	\begin{equation}
		\kappa\in\overline{\K}\mapsto\omega(\kappa)\in\C,\qquad\kappa\in\K\mapsto \G(\kappa)\in\mathcal{M}_n(\C),
	\end{equation}
	which are analytic in $\K$ and continuous on $\overline{\K}$.\footnote{Here $\overline{\mathcal{K}}$ is the closure of the set $\mathcal{K}$; the reader should not confuse this symbol with the one for complex conjugation. No ambiguities will arise from this choice. Finally, $\mathcal{M}_n(\mathbb{C})$ is the space of $n\times n$ matrices with complex entries.}
\end{hypothesis}
These are quite natural requests. Any polynomial function is naturally extended to an entire function and thus satisfies this assumption with $\mathcal{K}=\mathbb{C}$. More generally, an \textit{algebraic function} (e.g. a rational function or a fractional power of a polynomial) can be extended to a complex function which may have either isolated singularities in the complex plane and/or discontinuities in correspondence to curves in the complex plane that connect its \textit{branch points}, i.e. branch cuts. In this case, $\mathcal{K}$ will be taken as the complex plane minus a neighborhood of both the isolated singularities and the branch cuts. A concrete case will be discussed in Subsection~\ref{subsec:wqed}.
	
As a first consequence, both functions $k\in\R\mapsto\omega(k)\in\R$ and $k\in\R\mapsto \G(k)\in\mathcal{M}_n(\C)$ admit derivatives of all orders, each satisfying
	\begin{equation}
	\omega^{(\nu)}(-k)=(-1)^\nu\omega^{(\nu)}(k),\qquad \G^{(\nu)}(-k)=(-1)^\nu \G(k)^\intercal
	\label{eq:sim2bis}
	\end{equation}
	for all $k\in\R$ and $\nu\in\mathbb{N}$; therefore, by Hypothesis~\ref{hyp3},
	\begin{enumerate}
		\item[(i)] for all $\kappa$ such that $\kappa,\overline{\kappa},-\kappa,-\overline{\kappa}\in\K$, the functions $\omega(\kappa)$ and $\G(\kappa)$ satisfy
		\begin{equation}
		\overline{\omega(\kappa)}=\omega(\overline{\kappa}),\qquad \G(\kappa)^\dag=\G(\overline{\kappa})\label{eq:sim5}
		\end{equation}
		and
		\begin{equation}
		\omega(-\kappa)=\omega(\kappa),\qquad \G(-\kappa)=\G(\kappa)^\intercal.\label{eq:sim6}
		\end{equation}
		\item[(ii)] $\K$ can be always extended as a symmetric region with respect to both axes, in which~\eqref{eq:sim5} and~\eqref{eq:sim6} hold for all $\kappa$.
		\end{enumerate}
Indeed, given $\kappa=k+\i\eta\in \K$, via a series expansion and by applying~\eqref{eq:sim2bis} we obtain
	\begin{equation}
	\G(-k-\i\eta)=\sum_{\nu\in\mathbb{N}}\frac{1}{\nu!}\G^{(\nu)}(-k)(-\i\eta)^\nu=\sum_{\nu\in\mathbb{N}}\frac{1}{\nu!}\G^{(\nu)}(k)^\intercal(\i\eta)^\nu=\G(k+\i\eta)^\intercal;
	\end{equation}
	\begin{equation}
	\G(k+\i\eta)^\dag=\sum_{\nu\in\mathbb{N}}\frac{1}{\nu!}\G^{(\nu)}(k)^\dag(-\i\eta)^\nu=\sum_{\nu\in\mathbb{N}}\frac{1}{\nu!}\G^{(\nu)}(k)(-\i\eta)^\nu=\G(k-\i\eta),
	\end{equation}
	the argument for $\omega$ being analogous. $(ii)$ is an immediate consequence of $(i)$.
	
As we will see, the behavior of the solutions of the equation $\omega(\kappa)=z$ will crucially affect the properties of the self-energy. By the implicit function theorem, we know that, given any $\kappa_0\in\K$ such that $\omega'(\kappa_0)\neq0$, $\omega(\kappa)$ is \textit{locally invertible} near $\omega(\kappa_0)$, i.e. there is a neighborhood $S$ of $\omega(\kappa_0)$ and a function $\hat{\kappa}_{0}:S\mapsto\C$ such that 
\begin{equation}
\omega(\hat{\kappa_0}(z))=z\qquad\forall z\in S.
\end{equation}
Points with zero derivative and their corresponding values of $\omega(\kappa)$ are referred to as \textit{critical points} and \textit{critical values}, since $\omega(\kappa)$ is not locally invertible near them.

We will finally state here some simple properties of the solutions. The solutions $\kappa_0$ of the equation $\omega(\kappa)=z$ satisfy the following properties for all $z$:
	\begin{itemize}
		\item $\omega(\kappa_0)=z$ iff $\omega(-\kappa_0)=z$;
		\item $\omega(\kappa_0)=z$ iff $\omega(\overline{\kappa_0})=\overline{z}$;
		\item in particular, for all $E\in\R$,
		\begin{equation}\label{eq:quartet}
		\omega(\kappa_0)=E\text{  iff  }\omega(-\kappa_0)=E\text{  iff  }\omega(\overline{\kappa_0})=E\text{  iff  }\omega(-\overline{\kappa_0})=E.
		\end{equation}
	\end{itemize}
Some other generic considerations about this topic can be found in Appendix~\ref{app:solutions}.

\subsection{Behavior of the parameters at the boundary}
Finally, we will need a technical assumption about the behavior of $\G(\kappa)$ for large $|\kappa|$ and at the boundary $\Gamma$ of $\mathcal{K}$. Loosely speaking, we will require the matrix elements of $\G(\kappa)$ with $j\geq\ell$ not to grow ``too quickly'' neither for large values of $|\kappa|$ nor at the boundary of the region $\K$.

Here and in the following we will use the following notation: for any subset $S$ of the complex plane, we will write
\begin{equation}
	S^\pm=S\cap\C^\pm,
\end{equation}
where $\C^\pm=\left\{z\in\C:\,\pm\Im z>0\right\}$.
\begin{hypothesis}\label{hyp4}
For every $j\geq\ell$, there is $z_0\in\C$ such that the following conditions hold:
\begin{equation}
\lim_{R\to\infty}\,\max_{\kappa\in \K^+,\,|\kappa|=R}R\,\frac{|G_{j\ell}(\kappa)|}{|\omega(\kappa)-z_0|}=0,
\end{equation}
\begin{equation}
\left|\int_{\Gamma^+}\frac{G_{j\ell}(\kappa)}{\omega(\kappa)-z_0}\,\mathrm{d}\kappa\right|<\infty.
\label{eq:intbound}
\end{equation}
\end{hypothesis}
\begin{example}\label{ex3}
	Let us again consider the case $F_{j}(k)=F(k)\,\e^{\i kx_j}$ for some $x_1,\dots,x_n\in\R$, and thus 
	\begin{equation}
		G_{j\ell}(k)=G(k)\,\e^{\i k(x_j-x_\ell)},\qquad G(k)=|F(k)|^2.
	\end{equation}
	Without loss of generality, let us fix $x_1\leq x_2\leq\dots\leq x_n$ and let us only consider terms with $j\geq\ell$, so that we can equivalently write $G_{j\ell}(k)=G(k)\,\e^{\i k|x_j-x_\ell|}$. Suppose that the function $k\in\R\mapsto G(k)\in\R$ admits an analytic extension in some open connected $\mathcal{K}\subset\C$, satisfying
	\begin{equation}
	\lim_{R\to\infty}\,\max_{\kappa\in \K^+,\,|\kappa|=R}R\,\frac{|G(\kappa)|}{|\omega(\kappa)-z_0|}=0,\qquad
	\int_{\Gamma^+}\frac{|G(\kappa)|}{|\omega(\kappa)-z_0|}\,\mathrm{d}\kappa<\infty;
	\label{eq:scalar}
	\end{equation}
	then the function
	\begin{equation}
		\kappa\in\K\mapsto\G(\kappa)\in\mathcal{M}_n(\C),\qquad G_{j\ell}(\kappa)=G(\kappa)\,\e^{\i\kappa|x_j-x_\ell|}
	\end{equation}
	is an analytic extension of $k\in\R\mapsto\G(k)\in\mathcal{M}_n(\C)$. Besides, for every $\kappa\in\mathcal{K}$,
	\begin{equation}
|G_{j\ell}(\kappa)|=|G(\kappa)|\,\e^{-|x_j-x_\ell|\Im\kappa};
	\end{equation}
	since $x_j\geq x_\ell$ for every $j\geq\ell$, we have $|G_{j\ell}(\kappa)|\leq|G(\kappa)|$ and hence $\mathsf{G}(\kappa)$ satisfies the required hypotheses. We refer again to Subsection~\ref{subsec:wqed} for a concrete case in which such assumptions are fulfilled.
	 
\end{example}

\subsection{Decomposition of the self-energy}
We can now find an expression for the self-energy matrix.
\begin{proposition}\label{prop:prop}
	Consider a Friedrichs-Lee model such that Hypotheses~\ref{hyp1}--\ref{hyp4} are satisfied, and let $\mathcal{A}\subset\C$ be an open connected subset of $\C$ with $\A\cap\omega(\Gamma)=\emptyset$ and such that, for all $z\in\mathcal{A}$, the equation
	\begin{equation}
		\omega(\kappa)=z
	\end{equation}
	admits $2r$ continuous solutions $\pm\hat{\kappa}_1(z),\dots,\pm\hat{\kappa}_r(z)\in\K$, labeled in such a way that $\hat{\kappa}_s(z)\in\K^+$ for $z\in\A^+$. Then:
	
	$(i)$ for all $z\in\A^+\setminus\omega(\mathcal{C})$, where $\mathcal{C}=\{\kappa\in\K:\omega'(\kappa)=0\}$, the self-energy matrix $\S(z)$ can be written as such:
		\begin{equation}
		\label{eq:selfinf}
			\S(z)=\mathsf{b}(z)+2\pi\i\sum_{s=1}^r\mathsf{Z}(\hat{\kappa}_s(z)),\qquad\S(\overline{z})=\S(z)^\dag,
		\end{equation}
		where, for every $j\geq\ell$,
		\begin{equation}
b_{j\ell}(z)=\int_{\Gamma^+}\frac{G_{j\ell}(\kappa)}{\omega(\kappa)-z}\,\mathrm{d}\kappa,\qquad Z_{j\ell}(\kappa)=\frac{G_{j\ell}(\kappa)}{\omega'(\kappa)},
\label{eq:bjl}
		\end{equation}
		and $b_{\ell j}(z)=b_{j\ell}(z)$, $Z_{\ell j}(\kappa)=Z_{j\ell}(\kappa)$.
	
$(ii)$ Let $\A\cap\R\neq\emptyset$ and take $E\in\A\cap\R$. Defining
	\begin{eqnarray}
		I^+(E)&=&\left\{s=1,\dots,r:\,\Im\hat{\kappa}_s(E)>0\right\},\\I^0(E)&=&\left\{s=1,\dots,r:\,\Im\hat{\kappa}_s(E)=0\right\},
	\end{eqnarray}
	then the hermitian and anti-hermitian components of $\S(E\pm\i0)$ are
		\begin{equation}
		\Re\S(E\pm\i0)=\mathsf{b}(E)+2\pi\i\left(\sum_{s\in I^+(E)}\Z(\hat{\kappa}_s(E))+\sum_{s\in I^0(E)}\Re\Z(\hat{\kappa}_s(E))\right);
		\end{equation}
		\begin{equation}
		\Im\S(E\pm\i0)=\pm2\pi\i\sum_{s\in I^0(E)}\Im\Z(\hat{\kappa}_s(E));
		\end{equation}
		and the discontinuity of the self-energy at the real axis reads
		\begin{equation}
			\S(E+\i0)-\S(E-\i0)=2\pi\i\sum_{s\in I^0(E)}\left(\Z(\hat{\kappa}_s(E))-\Z(\hat{\kappa}_s(E))^\dag\right).
		\end{equation}
		
		$(iii)$ The self-energy can be analytically continued from $\mathcal{A}^+$ to $\mathcal{A}^-$, and its analytical continuation reads, for all $z\in\A^-\setminus\omega(\mathcal{C})$,
		\begin{equation}
			\S^{(\mathrm{II})}(z)=\mathsf{b}(z)+2\pi\i\sum_{s=1}^r\Z\left(\hat{\kappa}_s(z)\right).
		\end{equation}
\end{proposition}

This statement follows from a long, but straightforward calculation, reported in Appendix~\ref{app:proof}, which makes use of our technical assumptions and some techniques of elementary complex analysis. Here we will discuss about the statement. The self-energy $\mathsf{\Sigma}(z)$ is decomposed into the sum of two terms:
\begin{itemize}
	\item a contour contribution;
	\item residue contributions corresponding to the solutions of the equation $\omega(\kappa)=z$;
\end{itemize}
in particular, when analyzing the boundary values $\mathsf{\Sigma}(E+\i0)$ of the self-energy on the real line, we can distinguish between real and nonreal solutions of the equation $\omega(\kappa)=E$. This distinction will turn largely useful when analyzing, in Section~\ref{sec:identical}, the case of a family of identical atoms; in such a case, the contribution of real solutions will be \textit{dominant}, with nonreal poles and the contour contribution only providing ``small" corrections.

Notice that we can write such an expression for the self-energy in the region $\mathcal{A}$ of the complex plane, excluding countably many values of $z$; since the self-energy is necessarily analytic away from the real line, the singularities in $\omega(\mathcal{C})$ are removable and its extension to the full region $\A$ is trivial. If we can find finitely many subsets $\mathcal{A}_1,\dots,\mathcal{A}_N\subset\mathbb{C}$ that cover the full complex plane up to sets of zero measure, the self-energy is completely determined.

Finally, we remark that, for $z\in\A^-$, the self-energy $\S(z)$ and its analytic continuation $\S^{(\mathrm{II})}(z)$ from $\A^+$ differ by all and only the contributions of all those poles $\hat{\kappa}_s(z)$ that cross the real line when $z$ does. If all poles are nonreal for $z=E\in\mathbb{R}$, the two quantities coincide and the self-energy is continuous on it.

\section{Identical atoms in the infinite line}\label{sec:identical}
In this section we will apply the analysis of Section~\ref{sec:self} to a particular class of cases of remarkable physical importance. Following the discussion in Examples~\ref{ex1}-\ref{ex3}, consider a family of real numbers $x_1,x_2,\dots,x_n\in\R$, fixing $x_1\leq x_2\leq\ldots\leq x_n$, and take
\begin{equation}
	F_j(k)=F(k)\e^{\i kx_j},
\end{equation}
for some function $F(k)$ satisfying $F(-k)=\overline{F(k)}$, and therefore
\begin{equation}
G_{j\ell}(k)=G(k)\,\e^{\i k(x_j-x_\ell)},
\end{equation}
with $G(k)=|F(k)|^2$ (therefore $G(-k)=G(k)$) admitting an analytic extension to the complex region $\K$ satisfying Eq.~\eqref{eq:scalar}. Notice that, while the coupling functions $F_j(k)$ depend on the particular choice of coordinate system, the matrix $\mathsf{G}(k)$, and thus the self-energy (as well as all physical properties of the model) only depend on the mutual distance between the atoms.

It is easy to show that this choice corresponds, physically, to the case of \textit{identical} atoms. Indeed, consider a family of $n$ identical atoms in a one-dimensional space whose centers are placed at the positions $x_1,x_2,\ldots,x_n$, and let $\hat{F}_j(x)$ be the coupling between the $j$th atom and the field in the \textit{position} space. Since the atoms are identical, only differing each other by their \textit{position} $x_j$, there must be $\hat{F}(x)$ such that
\begin{equation}
\hat{F}_j(x)=\hat{F}(x-x_j).
\end{equation}
By performing a Fourier transform, we have
\begin{equation}
F_j(k)=F(k)\,\e^{\i kx_j},
\end{equation}
with $F_j$ and $F$ being the Fourier transforms of $\hat{F}_j$ and $\hat{F}$. Therefore
\begin{equation}
	G_{j\ell}(k)=|F(k)|^2\,\e^{\i k(x_j-x_\ell)}\equiv G(k)\,\e^{\i k(x_j-x_\ell)}.
\end{equation}
Finally, since $\hat{F}(x)$ is real-valued, $F(-k)=\overline{F(k)}$ and hence $G(-k)=G(k)$.

We shall hereafter focus on this particular case; accordingly, we shall also assume the excitation energies of all atoms to be equal, namely, $\varepsilon_1=\ldots=\varepsilon_n\equiv\varepsilon$.

\subsection{Dominant and suppressed terms in the self-energy}
Let us point out some properties of this system which follow easily from Proposition~\ref{prop:prop}: 
\begin{itemize}
	\item[(i)] the self-energy of this system can be written as such for all $z\in\A^+$:
\begin{equation}
\label{eq:selfid}
\S(z)=\mathsf{b}(z)+2\pi\i\sum_{s=1}^r Z(\hat{\kappa}_s(z))\mathsf{\Phi}_n(\hat{\kappa}_s(z)),\qquad\S(\overline{z})=\S(z)^\dag,
\end{equation}
where we define the matrices $\mathsf{\Phi}(\kappa)$, $\mathsf{b}(z)$ and the scalar quantity $Z(\kappa)$ via 
\begin{equation}
	[\mathsf{\Phi}_n(\kappa)]_{j\ell}=\e^{\i\kappa|x_j-x_\ell|},\quad b_{j\ell}(z)=\int_{\Gamma^+}\frac{G(\kappa)}{\omega(\kappa)-z}\,\e^{\i\kappa|x_j-x_\ell|}\,\mathrm{d}\kappa,\quad Z(\kappa)=\frac{G(\kappa)}{\omega'(\kappa)}.
\end{equation}
In particular,
\begin{equation}\label{eq:eq}
\S(E+\i0)=\mathsf{b}(E)+2\pi\i\left(\sum_{s\in I^+(E)} Z(\hat{\kappa}_s(E))\mathsf{\Phi}_n(\hat{\kappa}_s(E))+\sum_{s\in I^0(E)} Z(\hat{\kappa}_s(E))\mathsf{\Phi}_n(\hat{\kappa}_s(E))\right),
\end{equation}
\begin{equation}
\S(E-\i0)=\mathsf{b}(E)+2\pi\i\left(\sum_{s\in I^+(E)} Z(\hat{\kappa}_s(E))\mathsf{\Phi}_n(\hat{\kappa}_s(E))+\sum_{s\in I^0(E)} Z(\hat{\kappa}_s(E))\mathsf{\Phi}_n(\hat{\kappa}_s(E))^\dag\right).
\end{equation}
The same expression holds for the analytic continuation of the self-energy from $\A^+$ to $\A^-$:
\begin{equation}
\label{eq:selfidres}
\S^{(\mathrm{II})}(z)=\mathsf{b}(z)+2\pi\i\sum_{s=1}^r Z(\hat{\kappa}_s(z))\mathsf{\Phi}_n(\hat{\kappa}_s(z)),\qquad\S(\overline{z})=\S(z)^\dag.
\end{equation}
\item[(ii)] For every $j,\ell=1,\dots,n$ and every $s\in I^+(E)$,
\begin{equation}
	|b_{j\ell}(E)|\leq |b(E)|\e^{-m|x_j-x_\ell|},\qquad \Bigl|[\mathsf{\Phi}_n(\hat{\kappa}_s(E))]_{j\ell}\Bigr|\leq\e^{-\Im\hat{\kappa}_s(E)|x_j-x_\ell|},
\end{equation}
where
\begin{equation}
m=\min_{\kappa\in\Gamma^+}\Im\kappa,\qquad	b(E)=\int_{\Gamma^+}\frac{G(\kappa)}{\omega(\kappa)-E}\,\mathrm{d}k.
\end{equation}
\end{itemize}
Indeed, $(i)$ is an immediate consequence of Prop.~\ref{prop:prop}, while $(ii)$ (which will be crucial in the following) holds because of our choice of $\mathsf{G}(\kappa)$. We will refer to $\mathsf{\Phi}_n(\kappa)$ as the \textit{phase matrix}. 

Now, recall that the stable states of Friedrichs-Lee models are characterized by an energy value $E\in\R$ and an $\bm{a}\in\mathbb{C}^n$ satisfying
\begin{equation}
	\det\bigl(\mathcal{E}-E\,\mathsf{Id}-\S(E+\i0)\bigr)=0,\qquad\left(\mathcal{E}-E\,\mathsf{Id}-\S(E+\i0)\right)\bm{a}=\bm{0}.
	\label{eq:eig}
\end{equation}
Assuming their excitation energies to be equal, i.e. $\varepsilon_1=\ldots=\varepsilon_n\equiv\varepsilon$ for some $\varepsilon\in\R$, the conditions in Eq.~\eqref{eq:eig} yield an \textit{$E$-dependent eigenproblem} for the matrix $\S(E+\i0)$:
\begin{equation}
\det\bigl((\varepsilon-E)\mathsf{Id}-\S(E+\i0)\bigr)=0,\qquad\left((\varepsilon-E)\mathsf{Id}-\S(E+\i0)\right)\bm{a}=\bm{0}.
\label{eq:eig2}
\end{equation}
We must therefore search for eigenvalues of the self-energy and their corresponding eigenvectors. In this picture, the assertion $(ii)$ in the previous proposition will be of fundamental importance in the following: indeed, because of it, the \textit{off-diagonal} terms of all the contribution of the self-energy coming from following quantities:
\begin{itemize}
\item the matrix $\mathsf{b}(E)$, and
\item the residua of the poles $\hat{\kappa}_s(E)$ for $s\in I^+(E)$,
\end{itemize}
are \textit{exponentially suppressed}, in the following sense: if $m$ and, for all $s\in I^+(E)$, $\Im\hat{\kappa}_s(E)$ are ``sufficiently large", we may treat their off-diagonal components as a ``small" correction to the self-energy; their non-negligible diagonal terms will provide a shift in the value of $\varepsilon-E$.

Let us sum up these considerations and express the self-energy by stressing the distinction between suppressed terms and dominant ones. 
	Let $[E_1,E_2]\subset\mathbb{R}$ such that the quantity
	\begin{equation}
		m'=\min_{E\in[E_1,E_2],s\in I^+(E)}\Im\kappa_s(E)
	\end{equation}
	is strictly positive, and $M=\min\,\{m,m'\}$; also define
	\begin{equation}
		\alpha=\max_{E\in[E_1,E_2],s\in I^+(E)}\{b(E),Z(\hat{\kappa}_s(E))\}.
	\end{equation}
	Then, defining 
	\begin{equation}
	Z_{\mathrm{tot}}(E)=\sum_{s\in I^0(E)}Z(\hat{\kappa}_s(E)),\qquad z_s(E)=\frac{Z_s(\hat{\kappa}_s(E))}{Z_{\mathrm{tot}}(E)}
	\end{equation}
	and\footnote{Notice that, while not apparent from here, this quantity is real because of the symmetry properties of $Z(k)$.}
	\begin{equation}\label{eq:correction}
	\Delta(E)=b(E)+\i\sum_{s\in I^+(E)}Z(\hat{\kappa}_s(E)),
	\end{equation}
	the eigenvalue equation in Eq.~\eqref{eq:eig2} can be written as such:
	\begin{equation}
	\det\left\{\left(\sum_{s\in I^0(E)}z_s(\hat{\kappa}_s(E))\mathsf{\Phi}_n(\hat{\kappa}_s(E))+\mathsf{B}(E)\right)-\i\Lambda(E)\,\mathsf{Id}\right\}=0,
	\label{eq:eig4}
	\end{equation}
	where 
	\begin{equation}\label{lambda}
	\Lambda(E)=\frac{1}{Z_{\mathrm{tot}}(E)}\left(E-\varepsilon+\Delta(E)\right)
	\end{equation}
	and $\mathsf{B}(E)$ is a matrix with zero diagonal elements and whose off-diagonal elements satisfy
	\begin{equation}\label{eq:exp}
		|B_{j\ell}(E)|\leq\alpha\,\e^{-M|x_j-x_\ell|}.
	\end{equation}
This expression follows by substituting the expression~\eqref{eq:eq} for the self-energy in Eq.~\eqref{eq:eig2}, dividing by $Z_{\mathrm{tot}}(E)$, and applying the definitions of the quantities in the statement. As anticipated in Section~\ref{sec:self}, there is thus an explicit distinction between two types of quantities:
\begin{itemize}
	\item ``dominant" terms, associated to poles in $I^0(E)$ and the diagonal part of the integration on the contour $\Gamma^+$ and of the residua of poles in $I^+(E)$;
	\item ``suppressed" terms, associated to the off-diagonal part of the integration on the contour $\Gamma^+$ and of the residua of poles in $I^+(E)$,
\end{itemize}
all latter terms being grouped in the matrix $\mathsf{B}(E)$. We remark that the distinction between the ``dominant" and ``suppressed" terms ceases to hold when, in our desired energy range, there is some pole which becomes close to the real axis; this may happen whenever the interval $[E_1,E_2]$ is close to some critical value of $\omega$; when ``crossing" a critical value, the system may undergo a ``transition" and a pole becomes dominant or ceases to be dominant.

\subsection{Identical atoms with single dominant pole}\label{subsec:dominant}
As a particular case, let us consider a Friedrichs-Lee model for identical atoms with the following property: there is an interval $[E_0,E_1]\subset\R$ on which $\omega(k)=E$ admits a single real solution, {and thus the integrand in the definition of the self-energy admits one real pole, yielding a contribution to the self-energy which is \textit{dominant} in the sense discussed above. This happens, for instance, whenever the dispersion relation $k\in\R\mapsto\omega(k)\in\R$ is monotone in $\R^+$, this is the case for every $E$ in the range of $\omega$, i.e. for all $E>\omega(0)$. 

In this case, Eq.~\eqref{eq:eig4} reads
	\begin{equation}
\det\Bigl(\mathsf{\Phi}_n(\hat{\kappa}(E))+\mathsf{B}(E)-\i\Lambda(E)\,\mathsf{Id}\Bigr)=0.
\label{eq:eig5}
\end{equation}
As we will see, the spectral properties of the matrix $\mathsf{\Phi}_n(\kappa)$ can be studied explicitly: a recurrence relation for its characteristic polynomial, as well as an explicit expression for its determinant, will be evaluated in the next paragraph, in which we will also obtain exact results for the particular, yet instructive, case in which there are no suppressed terms in the self-energy. In the second and last paragraph, we will provide a brief discussion about the general case. Finally, a particular Friedrichs-Lee model belonging to the class of models examinated here, for which explicit expressions for all quantities can be obtained, will be briefly examined in Subsection~\ref{subsec:wqed}.

\paragraph{Simplest case: entire parameters, single couple of poles}\label{app:phase}
Let us start by examining the case in which the following conditions hold:
\begin{itemize}
	\item $\mathsf{b}(E)=\mathsf{0}$, i.e. there is no contour contribution. This is the case whenever the functions $\omega(k)$ and $G(k)$ admit analytic continuations in the \textit{whole} complex plane, i.e. $\omega(\kappa)$ and $G(\kappa)$ are \textit{entire} complex functions (e.g. they are polynomials);
	\item the equation $\omega(\kappa)=z$ admits a \textit{single} couple of solutions $\hat{\kappa}(z)$ in the whole complex plane,
\end{itemize}
and hence, for all $E>\omega(0)$, Eq.~\eqref{eq:eig5} holds with $\mathsf{B}(E)=\mathsf{0}$ and $\Delta(E)=0$ (see Eq.~\eqref{lambda} and~\eqref{eq:correction}), i.e.
	\begin{equation}
\det\left(\mathsf{\Phi}_n(\hat{\kappa}(E))-\i(E-\varepsilon)\,\mathsf{Id}\right)=0,
\label{eq:eig6}
\end{equation}
thus
\begin{equation}\label{eq:implicit}
E_j=\varepsilon-\i\,Z(\hat{\kappa}(E_j))\lambda_j\left(\hat{\kappa}(E_j)\right),
\end{equation}
with $\lambda_j(\kappa)$ being the $j$th eigenvalue of the phase matrix. As a result, the qualitative properties of the eigenvalues and resonances of the system will crucially depend on those of the phase matrix $\mathsf{\Phi}_n(\kappa)$. While Eq.~\eqref{eq:implicit} is implicit, if we suppose that the atom-boson coupling is \textit{weak}, heuristically we expect $Z(\hat{\kappa}(E_j))$ to be a ``small'' quantity, so that, via a series expansion and an iteration of Eq.~\eqref{eq:implicit}, we can take\footnote{This heuristic reasoning can be made more precise by rescaling the function $G(k)$, and thus $Z(k)$ as well, via a coupling constant $\beta>0$, and performing a series expansion in $\beta$: Eq.~\eqref{eq:explicit} then holds modulo $\mathcal{O}(\beta^2)$ corrections.}
\begin{equation}\label{eq:explicit}
E_j\sim\varepsilon-\i\,Z(\hat{\kappa}(\varepsilon))\lambda_j\left(\hat{\kappa}(\varepsilon)\right).
\end{equation}

It is therefore worth to take a closer look at the spectral properties of this matrix. Recall that
\begin{equation}
\left[\mathsf{\Phi}_n(\kappa)\right]_{j\ell}=\e^{\i\kappa|x_j-x_\ell|}=\begin{cases}
\e^{\i\kappa(x_j-x_\ell)},&j\geq\ell;\\
\e^{\i\kappa(x_\ell-x_j)},&\ell\geq j
\end{cases}
\end{equation}
This is an \textit{euclidean} one-dimensional matrix, i.e.\ is a function of the \textit{distances} between the set of points $\{x_j\}_{j=1,\dots,n}$ on the real line; besides, all the elements of its principal diagonal are $+1$, and hence its trace satisfies $\operatorname{tr}\mathsf{\Phi}_n(k,\eta)=n$.

Writing $\kappa=k+\i\eta$, with $k\in\R$ and $\eta\geq0$, and $\mathsf{\Phi}_n(k+\i\eta)\equiv\mathsf{\Phi}_n(k,\eta)$, the matrix can also decomposed in its hermitian and anti-hermitian components $\mathsf{\Phi}_n(k,\eta)=\mathsf{C}_n(k,\eta)+\i \mathsf{S}_n(k,\eta)$, where
\begin{equation}
[\mathsf{C}_n(k,\eta)]_{j\ell}=\e^{-\eta|x_j-x_\ell|}\cos k(x_j-x_\ell),\quad[\mathsf{S}_n(k,\eta)]_{j\ell}=\e^{-\eta|x_j-x_\ell|}\sin k|x_j-x_\ell|,
\end{equation}
both matrices being real and hermitian; their diagonal elements are respectively all $+1$ and all $0$, hence $\operatorname{tr}\mathsf{C}_n(k,\eta)=n$ and $\operatorname{tr}\mathsf{S}_n(k,\eta)=0$. We will refer to them as the \textit{cosine} and \textit{sine matrix}. It is easy to show that the cosine matrix $\mathsf{C}_n(k,\eta)$ is nonnegative for $\eta=0$ and positive for $\eta>0$, thus all eigenvalues of $\mathsf{\Phi}_n(\kappa)$ are characterized by
\begin{equation}
	\Re\lambda_j(\kappa)\geq0,\qquad j=1,\dots,n,
\end{equation}
and hence the corresponding resonances of the atom-field system will satisfy, in the weak coupling regime (see Eq.~\eqref{eq:explicit} and the related discussion),
\begin{equation}
	\Im E_j\sim- Z_{\mathrm{tot}}(\varepsilon)\Re\lambda_j(\hat{\kappa}(\varepsilon))\leq0,
\end{equation}
compatibly with the general fact that resonances always lie in the lower half-plane. 

A quantitative study of the spectral properties of the system requires the evaluation of the eigenvalues $\lambda_j(\kappa)$; in fact, a recurrence formula for the characteristic polynomial $p_n(\lambda;k,\eta)$ of $\mathsf{\Phi}_n(k,\eta)$ can be found, whence its eigenvalues can be conveniently computed. 
\begin{proposition}\label{prop:prop2}
	The characteristic polynomial of $\mathsf{\Phi}_n(\kappa)$ satisfies the recurrence system
	\begin{equation}
		\begin{cases}
		p_1(\lambda;\kappa)=1-\lambda;\\
		p_2(\lambda;\kappa)=(1-\lambda)^2-\e^{2\i\kappa(x_2-x_{1})};\\
		p_n(\lambda;\kappa)=\left[(1-\lambda)-(1+\lambda)\e^{2\i\kappa(x_n-x_{n-1})}\right]p_{n-1}(\lambda;\kappa)-\lambda^2\e^{2\i\kappa(x_n-x_{n-1})}p_{n-2}(\lambda;\kappa),
		\end{cases}
		\label{recurr}
	\end{equation}
	and
	\begin{equation}
		\det\mathsf{\Phi}_n(\kappa)=p_n(0,\kappa)=\prod_{j=1}^{n-1}[1-\e^{2\i\kappa(x_{j+1}-x_{j})}].
		\label{det}
	\end{equation}
\end{proposition}
The proof of this statement is reported in Appendix~\ref{app:proof2}. Notice that, for real $k$, we may also write
\begin{equation}
\det\mathsf{\Phi}_n(k,0)=(-2\i)^{n-1}\e^{\i k(x_n-n_{n-1})}\prod_{j=1}^{n-1}\sin k(x_{j+1}-x_j).
\end{equation}
\begin{figure}
\centering
\captionsetup{width=\linewidth}
\subfigure[$x=0.1;$]{\includegraphics[width=0.425\linewidth]{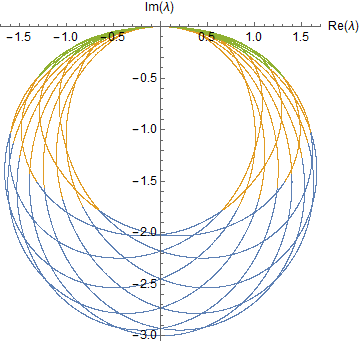}}\qquad
\subfigure[$x=0.2;$]{\includegraphics[width=0.425\linewidth]{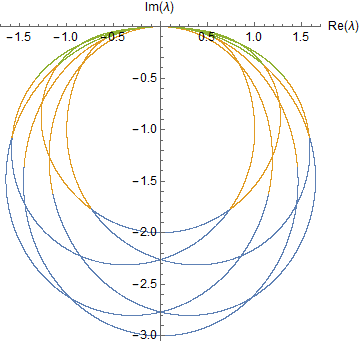}}\\
\subfigure[$x=0.3;$]{\includegraphics[width=0.425\linewidth]{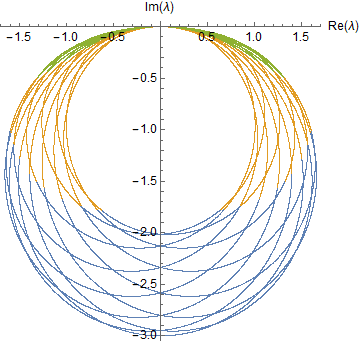}}\qquad
\subfigure[$x=0.4;$]{\includegraphics[width=0.425\linewidth]{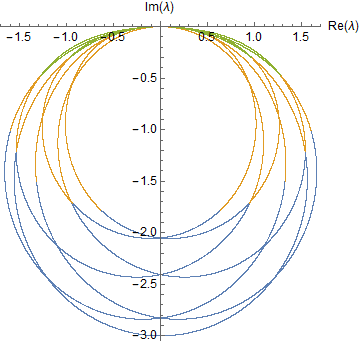}}\\
\subfigure[$x=0.5;$]{\includegraphics[width=0.425\linewidth]{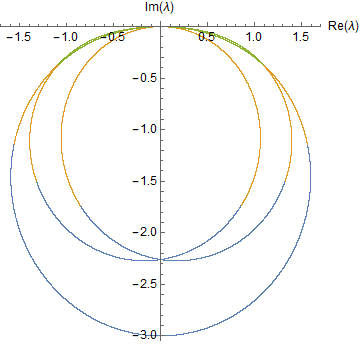}}\qquad
\subfigure[$x=1/\sqrt{2}.$]{\includegraphics[width=0.425\linewidth]{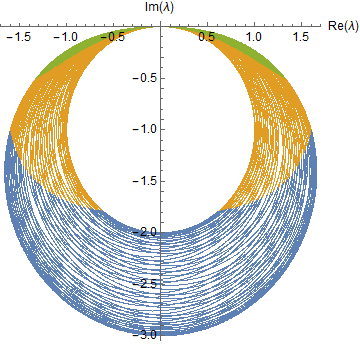}}\\
\caption{Eigenvalues of the matrix $-\i\mathsf{\Phi}_3(k,0)$ in the complex plane for various values of the ratio $x=\frac{x_2-x_1}{x_3-x_2}$ and $k$ ranging  in $[0,200]$; the different colors of the points correspond to distinct determinations of the roots of the polynomial $p_3(\lambda,k)$.
}
\label{fig:eigen}
\end{figure}
\begin{figure}
\centering
\captionsetup{width=\linewidth}
\subfigure[$x=0.1,\,\eta/\pi=0;$]{\includegraphics[width=0.425\linewidth]{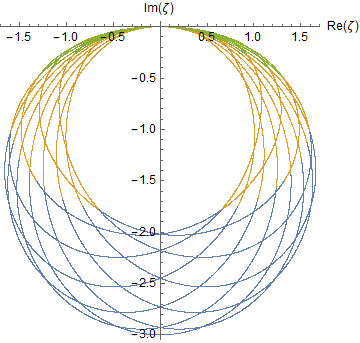}}\qquad
\subfigure[$x=0.1,\eta/\pi=0.1;$]{\includegraphics[width=0.425\linewidth]{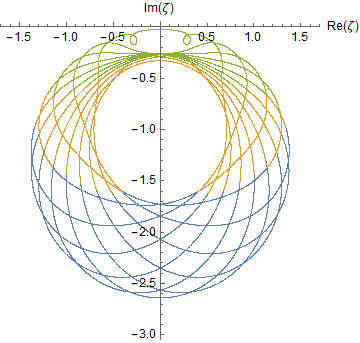}}\\
\subfigure[$x=0.1,\eta/\pi=0.3;$]{\includegraphics[width=0.425\linewidth]{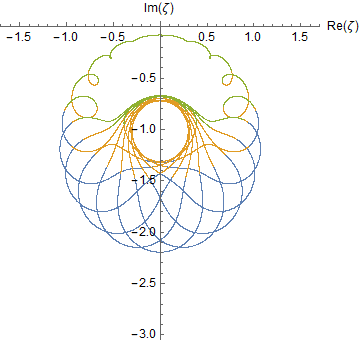}}\qquad
\subfigure[$x=0.1,\eta/\pi=0.4;$]{\includegraphics[width=0.425\linewidth]{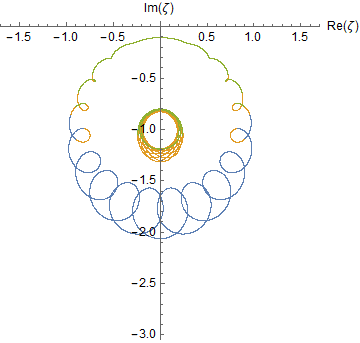}}\\
\subfigure[$x=0.1,\eta/\pi=0.5;$]{\includegraphics[width=0.425\linewidth]{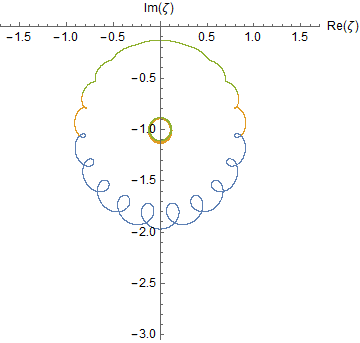}}\qquad
\subfigure[$x=0.1,\eta/\pi=1.0$]{\includegraphics[width=0.425\linewidth]{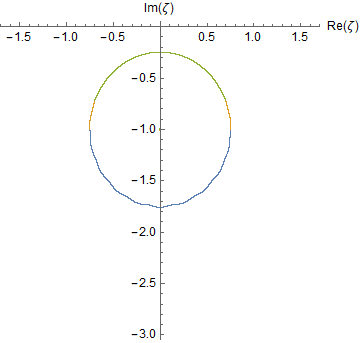}}\\
\caption{Eigenvalues of the matrix $-\i\mathsf{\Phi}_3(k,\eta)$ in the complex plane with $x=0.1$, $k$ ranging  in $[0,200]$, and various values of $\eta$; the different colors of the points correspond to distinct determinations of the roots of the polynomial $p_3(\lambda,k+\i\eta)$.
}
\label{fig:eigen2}
\end{figure}
\begin{figure}
	\centering
	\captionsetup{width=\linewidth}
	\subfigure[$x=1/3,\,x'=1/3;$]{\includegraphics[width=0.425\linewidth]{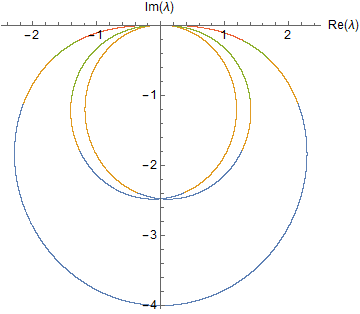}}\qquad
	\subfigure[$x=1/4,\,x'=3/4;$]{\includegraphics[width=0.425\linewidth]{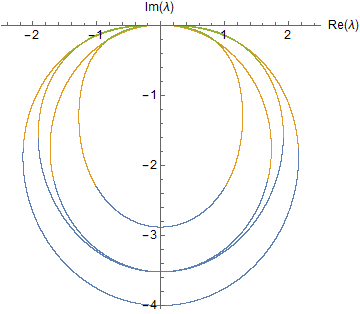}}\\
	\subfigure[$x=0.1,\,x'=0.4;$]{\includegraphics[width=0.425\linewidth]{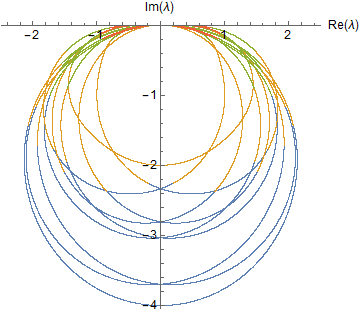}}\qquad
	\subfigure[$x=0.1,\,x'=0.5;$]{\includegraphics[width=0.425\linewidth]{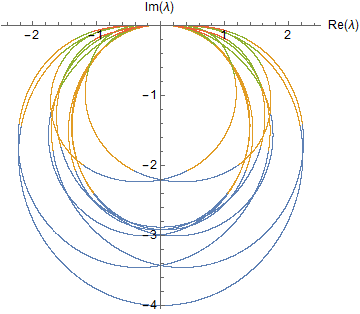}}\\
	\subfigure[$x=0.2,\,x'=1/\sqrt{7};$]{\includegraphics[width=0.425\linewidth]{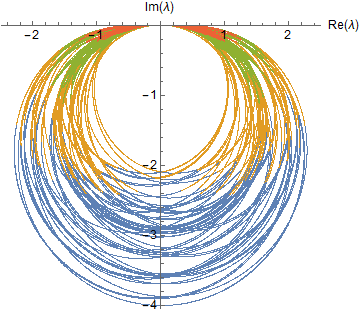}}\qquad
	\subfigure[$x=1/\sqrt{5},\,x'=1/\sqrt{7}.$]{\includegraphics[width=0.425\linewidth]{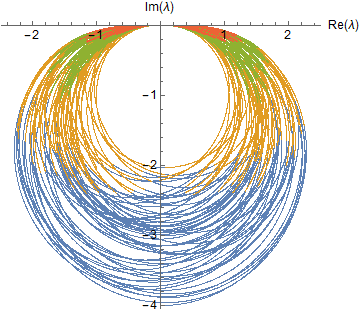}}\\
	\caption{Eigenvalues of the matrix $-\i\mathsf{\Phi}_4(k,0)$ in the complex plane for various values of the ratios $x=\frac{x_2-x_1}{x_4-x_1}$ and $x'=\frac{x_3-x_2}{x_4-x_1}$, and $k$ ranging in $[0,200]$; the different colors of the points correspond to distinct determinations of the roots of the polynomial $p_4(\lambda,k)$.
	}
	\label{fig:eigen3}
\end{figure}

Varying $\kappa$, we can thus study the spectrum of $\mathsf{\Phi}_n(\kappa)$, and hence the eigenvalues and resonances of the system, for every value of the parameters $x_1,\dots,x_n$; see Figs.~\ref{fig:eigen}--\ref{fig:eigen3}, where we plot the eigenvalues of $-\i\mathsf{\Phi}_n(k,\eta)$ for fixed chosen values of $\eta$ and $x_1,\dots,x_n$:
\begin{enumerate}
\item for any fixed value of $\eta\geq0$ and $x_1,\dots,x_n$, varying $k\in\mathbb{R}$ the eigenvalues describe trajectories in the lower complex half-plane that are closed \textit{if and only if} all the ratios between the distances $x_2-x_1$, $x_3-x_2$, ..., $x_n-x_{n-1}$ are in $\mathbb{Z}$; if the ratio between some couple of distances is irrational, the trajectories will fill a \textit{dense} region of the lower half-plane.
	\item for $\eta=0$, the eigenvalues of $-\i\mathsf{\Phi}_n(k,0)$ are contained inside a region of the complex plane with $0\leq\operatorname{Re}\lambda\leq n$ which is tangent to the imaginary axis at the origin $\lambda=0$, and outside the unit circle centered in $\lambda=1$. In particular, $\lambda=0$ is the \textit{only} eigenvalue with $\operatorname{Re}\lambda=0$.
	\item for $\eta>0$, the region in which the eigenvalues are contained is progressively ``shrunk" and asymptotically collapses, as $\eta\to\infty$, in the point $\lambda=1$.
	\end{enumerate}
	Before giving a brief explanation of these properties, as well as their relation to the bound states and resonances of the related Friedrichs-Lee model, let us introduce the following nomenclature. We will say that $k_0\in\mathbb{R}$ \textit{resonates} with the $j$th and the $\ell$th atom ($j>\ell$) if there is $\nu\in\mathbb{N}$ such that
	\begin{equation}\label{eq:ressing}
k_0=\frac{\nu\pi}{x_j-x_\ell},\qquad\nu\in\mathbb{N};
	\end{equation}
	if $k_0$ satisfies Eq.~\eqref{eq:ressing}, then $[\mathsf{\Phi}_n(k_0,0)]_{j\ell}=(-1)^{\nu+1}$. Clearly, if $k_0$ is resonating for the couple $(j,\ell)$, then for all $\nu'\in\mathbb{N}$ the value $\nu'k_0$ is also resonating for the same couple. Besides, $k_0$ can resonate with \textit{two} couples of atoms (say, the couple $(j,\ell)$ and the couple $(j',\ell')$) if and only if the ratio between $x_j-x_\ell$ and $x_{j'}-x_{\ell'}$ is a rational number, i.e. there are $p,q\in\mathbb{N}$ such that
	\begin{equation}
	\frac{x_j-x_\ell}{x_{j'}-x_{\ell'}}=\frac{p}{q};
	\end{equation}
	the same property holds for three of more couples of atoms.
	
Now, property $(i)$ simply follows from the fact that, for every $n\in\mathbb{N}$, the characteristic polynomial $p_n(\lambda,\kappa)$ is composed of finitely many terms $\e^{2\i\kappa(x_2-x_1)},\ldots,\e^{2\i\kappa(x_n-x_{n-1})}$; in this case, given the equations
	\begin{equation}\label{eq:res}
	k(x_{j+1}-x_j)=2\nu_j\pi,\qquad \nu_j\in\mathbb{N},\qquad j=1,\dots,n-1,
	\end{equation}
	then, if all the distances between atoms are commensurable, there is some $k_0\in\mathbb{R}$ such that $k=\nu' k_0$ solves Eq.~\eqref{eq:res} for all $\nu'\in\mathbb{N}$; in other words, starting from $k=0$, the curve closes in itself at every $k=\nu' k_0$. If some of the ratios are not commensurable, the equation above will only be satisfied for $k=0$ and the curves never close. 
	
	As for $(ii)$ and $(iii)$, notice that the quadratic forms associated to the matrices $\mathsf{C}_n(k,\eta)$ and $\mathsf{S}_n(k,\eta)$ read, for all $\bm{u}\in\mathbb{C}^n$,
	\begin{eqnarray}\label{eq:realpart}
	\bm{u}^\dag\mathsf{C}_n(k,\eta)\bm{u}&=&\sum_{j=1}^n|u_j|^2+\sum_{j=1}^n\sum_{j\neq\ell=1}^{n}\e^{-\eta|x_j-x_\ell|}\cos k(x_j-x_\ell)\Re\,\overline{u_j}u_\ell;\\
			\bm{u}^\dag\mathsf{S}_n(k,\eta)\bm{u}&=&\sum_{j=1}^n\sum_{j\neq\ell=1}^{n}\e^{-\eta|x_j-x_\ell|}\sin k(x_j-x_\ell)\Re\,\overline{u_j}u_\ell.
		\end{eqnarray}
		Clearly, for a fixed vector $\bm{u}$, the maximum possible value for the quadratic form of $\mathsf{C}_n(k,\eta)$ is obtained when all the terms $\cos k(x_j-x_\ell)$ equal one; therefore the inequality is saturated for $k$ satisfying Eq.~\eqref{eq:res}, which always happens if $k=0$ and, as discussed above, also for countably many other values of $k$ in the case of rational distances between all couples of atoms. This implies
		\begin{equation}
		\mathsf{C}_n(k,\eta)\leq\mathsf{C}_n(0,\eta),
		\end{equation}
		where the inequality is understood in the matrix sense, and hence the imaginary part of all the eigenvalues of $-\i\mathsf{\Phi}_n(k,\eta)$ is bounded from below by the largest eigenvalue of $-\i\mathsf{C}_n(0,\eta)$, whose absolute value, as $\eta$ grows from $0$ to $+\infty$, decreases monotonically from $n$ to $1$. Bounds for the real part of the eigenvalues may be found similarly. Besides, the fact that the spectrum ``collapses" at the point $-\i$ as $\eta\to\infty$ simply follows from the fact that, in that limit, all off-diagonal-terms of the self-energy vanish exponentially.

	These abstract properties are crucial, for our purposes, by virtue of Eq.~\eqref{eq:explicit}, which yields (in the small-coupling regime) a direct relation between the the resonances of the model and the eigenvalues of the phase matrix at $k=\hat{\kappa}(\varepsilon)$. Because of that, the curves in Figs.~\ref{fig:eigen} and~\ref{fig:eigen3} represent, for fixed values of the parameters $x_1,\dots,x_n$ and up to a multiplicative coupling-dependent constant, \textit{all admissible values} for the distances between the resonances and $\varepsilon$. By changing $\varepsilon$, one changes $k=\hat{\kappa}(\varepsilon)$ and thus ``moves along the curves". In particular, the analysis above shows that the absolute value of the imaginary part of said resonances cannot be arbitrarily large, and can at most reach a value which increases with the number $n$ of atoms.

Finally, let us focus on the case in which the eigenvalues are real, which, again by Eq.~\eqref{eq:explicit}, corresponds in our model to the emergence of a real-valued resonance, and thus of a bound state in the continuum. The only possible \textit{real} eigenvalue of $-\i\mathsf{\Phi}_n(k,\eta)$, which is thus linked to a \textit{real} energy of the system and hence to a proper bound state of the system, is obtained for $\eta=0$ and corresponds to $\lambda=0$. This happens if and only $\eta=0$ and is there is $j=1,\dots,n$ and $\nu\in\mathbb{Z}$ such that
	\begin{equation}
		k=\frac{\nu\pi}{x_{j+1}-x_j}.
		\label{reson}
	\end{equation}
	The degeneracy of the zero eigenvalue corresponds to the number of values of $j$ for which Eq.~\eqref{reson} is satisfied, i.e. to the number of couples of atoms which the value of $k$ ``resonates" with; the corresponding eigenspace is spanned by all the vectors $e^{(j)}\in\mathbb{C}^n$ defined as such:
	\begin{equation}
		e^{j}_\ell=\delta_{j,\ell}+(-1)^{\nu+1}\delta_{j+1,\ell}.
	\end{equation}
	Physically, a stable configuration for the system is obtained whenever at least one couple of \textit{consecutive atoms} resonates with $k$. Two particular cases are the following:
	\begin{itemize}
\item only \textit{one} couple of atoms resonates with $k$; to fix ideas, let the resonating atoms be the first and second ones, with positions $x_1$ and $x_2$, and hence
\begin{equation}
k=\frac{\nu\pi}{x_2-x_1}
\end{equation}
for some $\nu\in\mathbb{N}$, while, for all other couples of consecutive atoms, the equality above does not hold for any $\nu'\in\mathbb{N}$. This happens if the distances between consecutive couples of atoms are all mutually irrational. In that case, the atom configuration associated to the corresponding eigenvalue will be, up to a global normalization constant,
\begin{equation}
\bm{a}=(1,(-1)^{\nu+1},0,0,\ldots,0).
\end{equation}
Therefore, all other atoms have zero survival amplitude: this is a two-atom state.
\item \textit{all} couples of atoms resonate with $k$, i.e. Eq.~\eqref{eq:res} is satisfied. As discussed above, this can only happen if the distances between consecutive couples of atoms are all mutually rational, for example, if the atoms are placed on a regular one-dimensional lattice, i.e. $x_j=(j-1)d$ for some $d>0$.
	\end{itemize}

\paragraph{General case and degeneracy lifting}
In general, the matrix $\mathsf{B}(E)$ and the correction $\Delta(E)$ in the definition~\eqref{lambda} of $\Lambda(E)$ cannot be assumed to be zero: in general, there are suppressed poles, or $\K\neq\mathbb{C}$ (and hence the contribution of the contour integral on $\partial\K^+$ is nontrivial), or both. If so, Eq.~\eqref{eq:eig5} can no longer be exactly interpreted as an $E$-dependent eigenvalue equation for the phase matrix $\mathsf{\Phi}_n(\kappa)$. However, since all terms in the matrix $\mathsf{B}(E)$ are exponentially small by Eq.~\eqref{eq:exp}, heuristically we may expect the results of the previous paragraph to be still a good approximation of the actual ones.

We want to remark, however, that in some cases the role of the additional matrix $\mathsf{B}(E)$ may be of fundamental importance. In the previous paragraph, we have shown that, if there is are two or more couples of consecutive atoms whose distances are commensurable, there will be some value $k_0$ of $k$ which \textit{resonates} with all couples and, as a result, $\mathsf{\Phi}_n(k_0,0)$ will admits $0$ as a degenerate eigenvalue. Consequently, if $\hat{k}(\varepsilon)=k_0$, the field-atom system will admit a \textit{degenerate} space of bound states in the continuum with energy $E=\varepsilon$. However, in the presence of a nonzero matrix $\mathsf{B}(E)$, this degeneracy is generally going to be ``lifted".

Here we will not attempt to study this phenomenon in full generality; instead, we will provide a remarkable example which was studied in greater detail in \cite{bic2,bic3}, where we refer for details. Consider a family of \textit{equally spaced} atoms, i.e. take 
\begin{equation}
x_j=(j-1)d,\qquad j=1,\dots,n,
\end{equation}
for some $d>0$. In this case, the phase matrix $\mathsf{\Phi}_n(k,0)$ admits $0$ as an $(n-1)$-times degenerate eigenvalue if and only if
\begin{equation}\label{eq:resarray}
	k=\frac{\nu\pi}{d},\qquad\nu\in\mathbb{N};
\end{equation}
for those values of momentum, all couples of atom resonate. By our previous discussion one readily shows that, \textit{in the absence of corrections} (i.e. if Eq.~\eqref{eq:eig6}), the field-atom system admits an $(n-1)$-dimensional eigenspace of stable states, with energy $E=\varepsilon$, corresponding to the only constraint
\begin{equation}\label{eq:deg1}
	\sum_{j=1}^n(-1)^{j}a_j=0\text{  if $\nu$ is odd}
\end{equation}
or
\begin{equation}\label{eq:deg2}
\sum_{j=1}^na_j=0\text{  if $\nu$ is even}.
\end{equation}
In the general case, Eq.~\eqref{eq:eig5} holds. If one simply discarded the contribution of the matrix $\mathsf{B}(E)$ in Eq.~\eqref{eq:eig5}, thus only keeping diagonal corrections, one would only gain a $\nu$-dependent shift of the value of $\varepsilon$ corresponding to the resonant values~\eqref{eq:resarray} of the momentum, but the degeneracy subspaces defined by either Eq.~\eqref{eq:deg1} or~\eqref{eq:deg2} would be unchanged; however, if the matrix $\mathsf{B}(E)$ is taken into account, the degeneracy will be broken, as it is explicitly shown in \cite{bic2} for the cases $n=3,4$. A more complete study of this phenomenon for a regular array of a generic number $n$ of quantum atoms is presented in \cite{bic3}, where the symmetry properties of the actual stable states are discussed in greater detail and the physical meaning of this phenomenon is addressed.

	\subsection{An example from waveguide QED}\label{subsec:wqed}
	
	We will conclude this section by briefly revising, at the light of the discussion in previous sections and the formalism developed, a concrete instance of Friedrichs-Lee model describing a family of two-level emitters interacting with an electromagnetic field in a one-dimensional geometry, e.g. a photonic waveguide; this case was extensively studied in Refs. \cite{bic1,bic2,bic3}. We will show that the general formalism developed here does indeed apply to this case.
	
	From first principles, such a physical scenario can be effectively described, under some approximations (see \cite[Appendix A]{bic1}), by a Friedrichs-Lee model with
	\begin{equation}
		\omega(k)=\sqrt{k^2+m^2}
	\end{equation}
	and $F_j(k)=F(k)\e^{\i kx_j}$, with
	\begin{equation}
		F(k)=\sqrt{\frac{\gamma}{2\pi}}\frac{1}{\sqrt[4]{k^2+m^2}},
	\end{equation}
	where $m>0$ plays the role of an effective mass, and $\gamma>0$ is a coupling constant.
	
	It is easy to verify that all assumptions of this work hold in this case. Both functions are even, and a direct check shows that
	\begin{equation}
		\int_{-\infty}^{\infty}\frac{|F(k)|^2}{\omega(k)}\;\mathrm{d}k=\frac{\gamma}{2\pi}\int_{-\infty}^{\infty}\frac{1}{k^2+m^2}\;\mathrm{d}k=\frac{\gamma}{2m}<\infty,
	\end{equation}
	which is equivalent to Eq.~\eqref{eq:normalization} being verified; consequently, Hypotheses~\ref{hyp1}--\ref{hyp2} hold.
	
	Let us discuss the analyticity properties of the functions. The function $\omega(k)$ does indeed admit an analytic continuation to the complex plane:
	\begin{equation}
		\kappa\in\mathbb{C}\setminus\Bigl((-\i\infty,-\i m]\cup[\i m,\i\infty)\Bigr)\mapsto\omega(\kappa)=\sqrt{\kappa^2+m^2}\in\mathbb{C},
	\end{equation}
	where, above and hereafter, $\sqrt{\zeta}$ is to be interpreted as the \textit{principal} determination of the square root, that is, $\sqrt{\zeta}=\sqrt{|\zeta|}\e^{\i\arg\zeta/2}$; an analytic continuation of $|F(k)|^2$ is obtained analogously. The two points $\zeta=\pm\i m$ are branch points for the function above, the intervals $[\pm\i m,\pm\i\infty)$ being branch cuts. Consequently, Hypothesis~\ref{hyp3} does indeed hold with $\mathcal{K}\subset\mathbb{C}$ being obtained by ``removing'' from the complex plane any neighborhood $I^\pm_\epsilon$ of the branch cuts, i.e. $\mathcal{K}_\epsilon=\mathbb{C}\setminus\{I^+_\epsilon,I^-_\epsilon\}$, with $\epsilon>0$ small enough, its boundary being $\Gamma_\epsilon=\partial\mathcal{K}_\epsilon$; in fact, we will eventually take $\epsilon\to0$.
	
	To verify Hypothesis~\ref{hyp4}, we need to ensure Eq.~\eqref{eq:scalar} to hold. To do that, notice that, for all $\eta,\epsilon>0$,
	\begin{equation}\label{eq:decay}
		\frac{|F(\i\eta\pm\epsilon)|^2}{|\omega(\i\eta\pm\epsilon)|}=\frac{\gamma}{2\pi}\frac{1}{(\eta+m)^2+\epsilon^2},
	\end{equation}
	and a direct check shows that both properties in Eq.~\eqref{eq:scalar} are thus satisfied with $z_0=0$.
	
	To compute the self-energy, we need to study the behavior of the solutions of the equation $\omega(\kappa)=z$. Taking $\mathcal{A}=\{z\in\mathbb{C}:\;\Re z>0\}$, the equation admits a single couple of solutions $\pm\hat{\kappa}(z)$, where
	\begin{equation}
		\hat{\kappa}(z)=\sqrt{z^2-m^2},
	\end{equation}
	so that, for all $z\in\mathcal{A}^+$, Prop.~\ref{prop:prop}(i) implies
	\begin{equation}
		\Sigma_{j\ell}(z)=\i\gamma\,\frac{\e^{\i|x_j-x_\ell|\sqrt{z^2-m^2}}}{\sqrt{z^2-m^2}}+b_{j\ell}(z),
	\end{equation}
	with the contour term being
	\begin{equation}\label{eq:bjlwqed}
		b_{j\ell}(z)=	\frac{\gamma}{2\pi}\int_{\Gamma_\epsilon^+}\frac{\e^{\i\kappa|x_j-x_\ell|}}{\sqrt{\kappa^2+m^2}\left(\sqrt{\kappa^2+m^2}-z\right)}\,\mathrm{d}\kappa.
	\end{equation}
	The equality above holds for all $\epsilon>0$; we are thus free to take the limit $\epsilon\to0$, in which the integral over $\Gamma^+_\epsilon$ becomes the difference between the integrals, on the upper branch cut $[\i m,\i\infty)$, of the left and right limits of the integrand in Eq.~\eqref{eq:bjlwqed}; that is,
	\begin{eqnarray}
		b_{j\ell}(z)&=&-\frac{\gamma}{\pi}\int_{m}^\infty\frac{z\,\e^{-\eta|x_j-x_\ell|}}{\sqrt{\eta^2-m^2}\left(z^2+\eta^2-m^2\right)}\,\mathrm{d}\eta.
	\end{eqnarray}
	Finally, near the real axis, the self energy is continuous for $0<E<m$ and discontinuous for $E>m$. In the first case, we have
	\begin{eqnarray}
		\Sigma_{j\ell}(E)=\gamma\,\frac{\e^{-|x_j-x_\ell|\sqrt{m^2-E^2}}}{\sqrt{m^2-E^2}}+b_{j\ell}(E),\qquad(0<E<m),
	\end{eqnarray}
	while, in the second case,
	\begin{eqnarray}
		\Sigma_{j\ell}(E)=\i\gamma\,\frac{\e^{\i|x_j-x_\ell|\sqrt{E^2-m^2}}}{\sqrt{E^2-m^2}}+b_{j\ell}(E),\qquad(m<E<\infty).
	\end{eqnarray}	
In the second case (and assuming $\varepsilon_1=\ldots=\varepsilon_n$), we are precisely in the case sketched in Subsection~\ref{subsec:dominant}: the dominant contribution to the self-energy is the one corresponding to the solution $\hat{\kappa}(E)=\sqrt{E^2-m^2}$ to the equation $\omega(k)=E$. By neglecting the contour contribution, all real eigenvalues (thus corresponding to bound states in the continuum) would be found at energies satisfying the resonance condition~\eqref{eq:ressing} for some $j$ (i.e., two nearby atoms \textit{resonate}), which, in terms of energies, translates into
\begin{equation}
	E=E_\nu,\qquad E_\nu=\sqrt{\frac{\nu^2\pi^2}{(x_{j+1}-x_j)^2}+m^2}\qquad (\nu\in\mathbb{N}),
\end{equation}
the corresponding eigenvectors having the same qualitative structure extensively discussed in the first paragraph of Subsection~\ref{subsec:dominant}. Taking into account the additional contour contributions, richer (and more complicated) phenomena, like the degeneracy lifting briefly discussed at the end of the previous subsection, enter the game, as already stated in the general case.

\section{Concluding remarks}
In this article, an explicit expression for the self-energy of Friedrichs-Lee models was evaluated under precise assumptions on the dispersion relation $\omega(k)$ of the boson field and the form factors $F_{j\ell}(k)$ modeling the atom-field interaction; the self-energy can be decomposed as the sum of finitely many residue pole terms (poles), each being linked to a (possibly complex) solution of the equation $\omega(k)=E$, plus a contour contribution. No specific choice of either $\omega(k)$ or $F_{j\ell}(k)$ is used in our calculations: the result is valid for a large class of choices of the aforementioned parameters, and is thus suitable to diverse physical scenarios.

We have also provided a detailed analysis of a subclass of such models which describe identical atoms in an infinite line; in this framework, the contribution of real poles will be \textit{dominant}, with nonreal poles and the contour contribution only providing exponentially suppressed corrections. In particular, the case in which there is a single real pole can be analyzed thoroughly: when neglecting the small corrections, the eigenvalue problem in this case has been solved exactly and a physical interpretation of the result was given. As an important \textit{caveat}, we have briefly commented about the role of such corrections, which, albeit small, can become of crucial importance when analyzing degenerate situations. We stress that the distinction between dominant and suppressed contributions to the self-energy could be, indeed, obtained with little effort under more general assumptions.

Finally, while our results have been obtained for a Friedrichs-Lee model characterized by a boson field confined in an infinite line (hence the momentum $k$ ranging on the full real line), e.g. an infinitely long waveguide, similar strategies can been applied as well in order to describe fields confined in other geometries (the simplest examples being a half-line or a cavity). In these cases, the problem will be enriched by the possibility of selecting different boundary conditions that may alter significantly the spectral properties of the system, and thus its dynamics. This will be examined in future works.

\section*{Acknowledgments}
We acknowledge discussions with Paolo Facchi. This work is partially supported by Istituto Nazionale di Fisica Nucleare (INFN) through the project “QUANTUM” and by the Italian National Group of Mathematical Physics (GNFM-INdAM).

\section*{Data availability statement}
Data sharing not applicable to this article as no datasets were generated or analysed during the current study.

\appendix

\section{Behavior of the poles}
\label{app:solutions}
We will provide here additional details about the dependence of the solutions $\pm\hat{\kappa}_s(z)$ on $z$, i.e., heuristically, how they ``move" in $\K$ when varying $z$ in some properly chosen subset $\A$ of the complex plane.	

We will address the problem in the following way. An analytic function $\kappa\in\K\mapsto\omega(\kappa)\in\C$ is uniquely associated to a couple of real-valued, real-differentiable functions\footnote{Here and in the following, with an abuse of notation, we will identify the complex number $\kappa=k+\i\eta\in\C$ with the real couple $(k,\eta)\in\R^2$, and hence $\C$ with $\R^2$.}
\begin{equation}
(k,\eta)\in\K\mapsto u(k,\eta)\in\R,\qquad (k,\eta)\in\K\mapsto v(k,\eta)\in\R
\end{equation}
such that the Cauchy-Riemann equations are satisfied for all $(k,\eta)\in\K$:
\begin{equation}\begin{cases}
\partial_ku(k,\eta)=\partial_\eta v(k,\eta);\\
\partial_\eta u(k,\eta)=-\partial_k v(k,\eta),
\end{cases}
\end{equation}
the correspondence being given by $\omega(k+\i\eta)=u(k,\eta)+\i v(k,\eta)$. By differentiating the previous system and introducing the notation $\bm{\nabla}=\partial_k\,\bm{i}+\partial_\eta\,\bm{j}$ and $\Delta=\bm{\nabla}\cdot\bm{\nabla}$, one obtains
\begin{equation}
\Delta u(k,\eta)=\Delta v(k,\eta)=0,\qquad\bm{\nabla}u(k,\eta)\cdot\bm{\nabla}v(k,\eta)=0,
\end{equation}
i.e. $u(k,\eta)$ and $v(k,\eta)$ are two harmonic scalar fields in $\K\subset\R^2$ whose gradients, which are well-defined in the whole region, are constrained to be orthogonal whenever nonzero. Besides, by definition, the complex derivative of $\omega(\kappa)$ in some $\kappa_0=k_0+\i\eta_0$ will be given by
\begin{equation}
\omega'(k_0+\i\eta_0)=\left(\bm{\nabla}u(k_0,\eta_0)+\i\bm{\nabla}v(k_0,\eta_0)\right)\cdot\bm{h}
\end{equation}
for any unit vector $\bm{h}\in\R^2$; the Cauchy-Riemann equations ensure the independence of this quantity of the direction $\bm{h}$. In particular, $\kappa_0$ is critical (i.e. $\omega'(\kappa_0)=0$) if and only if $\bm{\nabla}u(k_0,\eta_0)=\bm{\nabla}v(k_0,\eta_0)=\bm{0}$, and, as an immediate consequence of the the Cauchy-Riemann equations, $\bm{\nabla}u(k_0,\eta_0)=\bm{0}$ if and only if $\bm{\nabla}v(k_0,\eta_0)=\bm{0}$.

The behavior of the two fields can be conveniently represented through their \textit{equipotential lines}, i.e. the curves implicitly defined by the equations $u(k,\eta)=u_0$ and $v(k,\eta)=v_0$ for some $u_0,v_0\in\R$. Each noncritical point $\kappa_0=(k_0,\eta_0)$ is crossed by one equipotential line for $u(k,\eta)$, namely $u(k,\eta)=u(k_0,\eta_0)$, and one critical line for $v(k,\eta)$, namely $v(k,\eta)=v(k_0,\eta_0)$; by construction the two equipotential lines are respectively orthogonal to $\bm{\nabla}u(k_0,\eta_0)$ and $\bm{\nabla}v(k_0,\eta_0)$ and hence have the same direction as, respectively, $\bm{\nabla}v(k_0,\eta_0)$ and $\bm{\nabla}u(k_0,\eta_0)$. This also implies that the two sets of equipotential lines are orthogonal in every noncritical point.

The situation is different if we consider a critical point, since both gradients vanish in it. It is easy to show that $\kappa_0$ is a critical point of order $m$, i.e. with $\omega^{(s)}(\kappa_0)=0$ for all $s=1,2,\dots,m-1$ and $\omega^{(m)}(\kappa_0)\neq0$ if and only if $m$ equipotential lines for both $u(k,\eta)$ and $v(k,\eta)$, respectively with value $u(k_0,\eta_0)$ and $v(k_0,\eta_0)$, intersect in $(k_0,\eta_0)$; each set of equipotential lines form $m$ equal angles of $2\pi/m$, each one being bisected by an equipotental line of the other set. This simply follows by observing that, near the critical point, the behavior of $\omega(\kappa)$ is the same as the $m$th power:
\begin{equation}
\omega(\kappa)\sim\omega(\kappa_0)+\frac{1}{m!}\omega^{(m)}(\kappa_0)(\kappa-\kappa_0)^m,
\end{equation} 
whose equipotential lines near the critical point $\kappa_0$ behave exactly as described.

Interestingly, the behavior of the solutions of the equation $\omega(\kappa)=z$ can be conveniently described through equipotential lines. This is proven straightforwardly:
\begin{proposition}
	Let $\A$ an open connected subset of the complex plane and let $z\in\A\mapsto\hat{\kappa}_0(z)\in\K$ a solution of the equation $\omega(\kappa)=z$, i.e. $\omega(\hat{\kappa}_0(z))=z$ for all $z\in\A$. Then:
	\begin{itemize}
		\item fixing $E_0\in\R$ and given $I_{E_0}=\{\delta\in\R:\,E_0+\i\delta\in\A\}$, the function $\delta\in I_{E_0}\mapsto\hat{\kappa}_0(E_0+\i\delta)$ has values in an equipotential line $u(k,\eta)=E_0$;
		\item fixing $\delta_0\in\R$ and given $J_{\delta_0}=\{E\in\R:\,E+\i\delta_0\in\A\}$, the function $E\in J_{\delta_0}\mapsto\hat{\kappa}_0(E+\i\delta_0)$ has values in an equipotential line $v(k,\eta)=\delta$.
	\end{itemize}
\end{proposition}
\begin{proof}
	$(i)$ By definition we have $\omega(\hat{\kappa}_0(E_0+\i\delta))=E_0+\i\delta$ and hence $\Re\omega(\hat{\kappa}_0(E_0+\i\delta))=E_0$ for all $\delta$, i.e. $\hat{\kappa}_0(E_0+\i\delta)$ belongs to an equipotential line $u(k,\eta)=E_0$. $(ii)$ is proven identically.	
\end{proof}
In words, starting from some $z_0=E_0+\i\delta_0$ and varying either the imaginary or real part of $z$, the solutions of the equation will move along the equipotential lines. Besides, by our previous discussion, when $z$ approaches a critical value corresponding to a critical point of order $m$ (either by varying $E$ or $\delta$, or also following any other trajectory in the $z$ plane), $m$ simple solutions of the equation ``collide" in the critical point.

The proposition above has a simple consequence which we will use in the proof of the main theorem:
\begin{proposition}\label{rem:notcross}
Let $\mathcal{A}$ a subset of the complex plane and $\hat{\kappa}_0(z)$ a solution which is continuous in $\mathcal{A}$. If $\Im\hat{\kappa}_0(z_0)>0$ for some $z_0\in\mathcal{A}^+$, then $\Im\hat{\kappa}_0(z)>0$ for all $z\in\mathcal{A}^+$.
\end{proposition}
\begin{proof}
	A solution $\hat{\kappa}_0(z)$ cannot cross an equipotential line $u(k,\eta)=u_0$ (resp. $v(k,\eta)=v_0$) if $z$ ranges in a region with does not contain values with real part equal to $u_0$ (resp. imaginary part equal to $v_0$). In particular, since $\A^+$ has no intersection with the real line, then $\hat{\kappa}_0(z)$ will never cross the real line, since $\R$ is an equipotential line with $v(k,\eta)=0$. This means that, if $\Im\hat{\kappa}_0(z_0)>0$ for \textit{some} $z_0\in\A$, then $\Im\hat{\kappa}_0(z)>0$ for \textit{all} $z\in\A$.
\end{proof}

\section{Proof of Prop.~\ref{prop:prop}}\label{app:proof}

Before proving the result, a preliminary lemma gathering two immediate consequences of Hypothesis.~\ref{hyp4}, will be useful.
\begin{lemma}\label{lemma:b}$(i)$ 
	Let $z\in\C$ and $M,\epsilon>0$ such that $|\omega(\kappa)-z|>\epsilon$ for all $\kappa\in\K$ with $|\kappa|>M$. Then
	\begin{equation}
	\lim_{R\to\infty}\,\max_{\kappa\in \K^+,\,|\kappa|=R}R\,\frac{|G_{j\ell}(\kappa)|}{|\omega(\kappa)-z|}=0.
	\end{equation}
	
	$(ii)$ Let $z\in\C$ and $\epsilon>0$ such that $|\omega(\kappa)-z|>\epsilon$ for all $\kappa\in\Gamma^+$. Then, for all $j\geq\ell$, the quantity
	\begin{equation}
	b_{j\ell}(z)=
	\int_{\Gamma^+}\frac{G_{j\ell}(\kappa)}{\omega(\kappa)-z}\,\mathrm{d}\kappa
	\end{equation}
	is well-defined and satisfies $b_{j\ell}(\overline{z})=\overline{b_{j\ell}(z)}$; if $z=E\in\R$, $b_{j\ell}(E)\in\R$.
\end{lemma}
\begin{proof}
	Indeed, we have
	\begin{equation}
	\frac{1}{\omega(\kappa)-z}=\left(1+\frac{z-z_0}{\omega(\kappa)-z}\right)\frac{1}{\omega(\kappa)-z_0}
	\end{equation}
	and thus, if $|\omega(\kappa)-z|>\epsilon$,
	\begin{equation}
	\left|\frac{1}{\omega(\kappa)-z}\right|\leq\left(1+\frac{|z-z_0|}{|\omega(\kappa)-z|}\right)\left|\frac{1}{\omega(\kappa)-z_0}\right|\leq\left(1+\frac{|z-z_0|}{\epsilon}\right)\left|\frac{1}{\omega(\kappa)-z_0}\right|.
	\end{equation}
	By this inequality, the first property follows immediately and the second one follows by dominated convergence.\end{proof}

Let us finally prove Prop.~\ref{prop:prop}.
\begin{proof}[Proof of Prop.~\ref{prop:prop}]
	Since $\S(\overline{z})=\S(z)^\dag$ and $\S(z)^\intercal=\S(z)$, we only need to compute $\Sigma_{j\ell}(z)$ for $j\geq\ell$ and $z\in\A^+$.
	
	$(i)$ First of all, since $\mathcal{A}$ does not intersect $\omega(\Gamma)$, then for all $z\in\A$ $\mathsf{b}(z)$ is well-defined by Lemma~\ref{lemma:b} with $\epsilon=\inf_{\kappa\in\Gamma^+}|\omega(\kappa)-z|>0$; besides, the functions $z\in\mathcal{A}\mapsto\hat{\kappa}_s(z)\in\K$ are continuous in $\mathcal{A}$ and, moreover, analytic away from the set $\mathcal{C}$ of zeros of $\omega'(\kappa)$ (see the discussion in Appendix~\ref{app:solutions}).
	
	Now fix $z\in\A\setminus\omega(\mathcal{C})$. Since the self-energy is well-defined because of Hypothesis~\ref{hyp1}, we can compute it as
	\begin{equation}
	\S(z)=\lim_{R\to\infty}\int_{-R}^R\frac{\G(k)}{\omega(k)-z}\,\mathrm{d}k.
	\end{equation}
	By Hypothesis~\ref{hyp3}, the integrand has a meromorphic continuation in the open connected subset $\K$ of the complex plane, its singularities being the solutions of the equation $\omega(\kappa)=z$; since there are finitely many singularities, for sufficiently large $R$ they all lie inside the circle $C_R=\{\kappa\in\C:\;|\kappa|<R\}$. For such values of $R$, define
	\begin{equation}
	\K^\pm_R=\K^\pm\cap C^\pm_R,
	\end{equation}
	its boundary $\delta\K^\pm_R$ being a piecewise smooth curve which can be decomposed as
	\begin{equation}
	\delta\K^\pm_R=\Gamma^\pm_R\cup \delta C^\pm_R,
	\end{equation}
	where $\Gamma^\pm_R=\{\zeta\in\Gamma^\pm:\;|\zeta|<R\}$ and $\delta C^\pm_R$ is the union of a number of circular arcs with radius $R$; see Fig.~\ref{fig:tikz}. For each couple of solutions $\pm\hat{\kappa}_s(z)$, either they are both real or they belong to distinct half-planes; in particular (see Prop.~\ref{rem:notcross}), each solution belongs either to the positive or negative half-plane for all $z\in\A^+$.
	
	\begin{figure}[t]\centering
		\begin{tikzpicture}[decoration={markings,
			mark=at position 0.8cm with {\arrow[line width=1pt]{<}},
			mark=at position 3.35cm with {\arrow[line width=1pt]{<}},
		}
		]
		
		\draw[name path=A,help lines,->,gray,line width=0.5pt] (-3.6,0) -- (3.6,0) coordinate (xaxis);
		\draw[help lines,->,gray,line width=0.5pt] (0,-1) -- (0,3.3) coordinate (yaxis);
		
		\draw [->,red!80!black,line width=0.8pt] (0.8,0.15) arc (-90:360:0.25);
		
		\draw [->,red!80!black,line width=0.8pt] (-2.0,0.15) arc (-90:360:0.25);
		
		\draw [->,blue!80!black,line width=0.8pt] (-3.18,0) -- (-1.5,0);
		\draw [->,blue!80!black,line width=0.8pt] (-1.6,0) -- (1.5,0);
		\draw [blue!80!black,line width=0.8pt] (1.4,0) -- (3.18,0);
		
		\node [cross out,thick,draw=black,scale=0.5] at (0.8,0.4) {};
		\node [cross out,thick,draw=black,scale=0.5] at (-0.8,-0.4) {};
		
		\node [cross out,thick,draw=black,scale=0.5] at (2.0,-0.4) {};
		\node [cross out,thick,draw=black,scale=0.5] at (-2.0,0.4) {};
		
		\draw[name path=B,thick,color=black] plot [smooth] coordinates {(-3.6,2.91) (-1,1.5) (1,1.5) (3.6,2.91)};
		\draw[thick,color=red!80!black] plot [smooth] coordinates {(-2.4,2.07) (-1.1,1.4) (0,1.252) (1.1,1.4) (2.4,2.07)};
		\draw[->-,red!80!black,line width=0.8pt] (3.157,0) arc (0:41:3.157);
			\draw[->-,red!80!black,line width=0.8pt] (-2.4,2.07) arc (-41.3:0:-3.157);
		\draw[dotted,red!80!black,line width=0.8pt] (-2.4,2.07) arc (-41.3:-140:-3.177);
		
		\path[draw,black,dotted,line width=0.6pt] (0.8,0.4) -- (0.8,0);
		\path[draw,black,dotted,line width=0.6pt] (-0.8,-0.4) -- (-0.8,0);
		
		\path[draw,black,dotted,line width=0.6pt] (2.0,-0.4) -- (2.0,0);
		\path[draw,black,dotted,line width=0.6pt] (-2.0,0.4) -- (-2.0,0);
		
		\node[above,gray!60!black] at (3.7,0) {$\mathrm{Re}\,\kappa$};
		\node[below right,gray!60!black] at (0,3.1) {$\mathrm{Im}\,\kappa$};
		\node[below,black] at (0.8,0) {$\hat{\kappa}_1$};
		\node[above,black] at (-0.8,-0.05) {$-\hat{\kappa}_1$};
		
		\node[below,black] at (-2.0,0) {$-\hat{\kappa}_2$};
		\node[above,black] at (2.0,-0.05) {$\hat{\kappa}_2$};
		\node[black] at (-3,2.95) {$\Gamma^+$};
		\node[red!80!black] at (-0.3,0.97) {$\Gamma^+_R$};
		\node[red!80!black] at (-3.18,1.8) {$\delta C^+_R$};
		\node[red!80!black] at (3.15,1.8) {$\delta C^+_R$};
		\node[below,blue!50!black] at (-3.18,0) {$-R$};
		\node[below,blue!50!black] at (3.18,0) {$R$};
		
		\tikzfillbetween[of=A and B]{yellow, opacity=0.12};
		
		\end{tikzpicture}
		
		\begin{tikzpicture}[decoration={markings,
			mark=at position 0.8cm with {\arrow[line width=1pt]{<}},
			mark=at position 3.35cm with {\arrow[line width=1pt]{<}},
		}
		]
		
		\draw[name path=A,help lines,->,gray,line width=0.5pt] (-3.6,0) -- (3.6,0) coordinate (xaxis);
		\draw[help lines,->,gray,line width=0.5pt] (0,-1) -- (0,3) coordinate (yaxis);
		
		\draw [->,red!80!black,line width=0.8pt] (0.8,0.15) arc (-90:360:0.25);
		
		\draw [->,red!80!black,line width=0.8pt] (-2.0,0.15) arc (-90:360:0.25);
		
		\draw [->,blue!80!black,line width=0.8pt] (-3.6,0) -- (-1.5,0);
		\draw [->,blue!80!black,line width=0.8pt] (-1.6,0) -- (1.5,0);
		\draw [blue!80!black,line width=0.8pt] (1.4,0) -- (3.6,0);
		
		\node [cross out,thick,draw=black,scale=0.5] at (0.8,0.4) {};
		\node [cross out,thick,draw=black,scale=0.5] at (-0.8,-0.4) {};
		
		\node [cross out,thick,draw=black,scale=0.5] at (2.0,-0.4) {};
		\node [cross out,thick,draw=black,scale=0.5] at (-2.0,0.4) {};
		
		\draw[name path=B,thick,color=black] plot [smooth] coordinates {(-3.6,2.91) (-1,1.5) (1,1.5) (3.6,2.91)};
		\draw[thick,color=red!80!black] plot [smooth] coordinates {(-3.6,2.78) (-2.4,2.07) (-1.1,1.4) (0,1.252) (1.1,1.4) (2.4,2.07) (3.6,2.78)};
			
		\path[draw,black,dotted,line width=0.6pt] (0.8,0.4) -- (0.8,0);
		\path[draw,black,dotted,line width=0.6pt] (-0.8,-0.4) -- (-0.8,0);
		
		\path[draw,black,dotted,line width=0.6pt] (2.0,-0.4) -- (2.0,0);
		\path[draw,black,dotted,line width=0.6pt] (-2.0,0.4) -- (-2.0,0);
		
		\node[above,gray!60!black] at (xaxis) {$\mathrm{Re}\,\kappa$};
		\node[below right,gray!60!black] at (yaxis) {$\mathrm{Im}\,\kappa$};
		\node[below,black] at (0.8,0) {$\hat{\kappa}_1$};
		\node[above,black] at (-0.8,-0.05) {$-\hat{\kappa}_1$};
		
		\node[below,black] at (-2.0,0) {$-\hat{\kappa}_2$};
		\node[above,black] at (2.0,-0.05) {$\hat{\kappa}_2$};
		\node[black] at (-3,2.95) {$\Gamma^+$};
			
		\tikzfillbetween[of=A and B]{yellow, opacity=0.12};
		
		\end{tikzpicture}
		
		\caption{Integration contour for a finite value of $R$ (top) and in the limit $R\to\infty$ (bottom): the integrals on the blue and the red contours coincide. The yellow area is the region $\mathcal{K}^+$ of the complex plane.}		
		\label{fig:tikz}
	\end{figure}
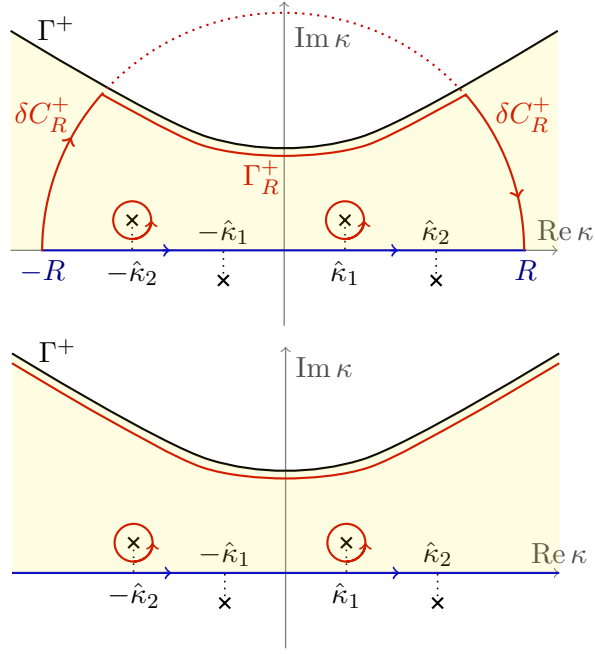\FloatBarrier
	By the residue theorem, the integral on the contour of the region $\K^+_R$ is obtained via the residua of the function $\omega(\kappa)$ in all solutions of the equation $\omega(\kappa)=z$ included in it, i.e. $+\hat{\kappa}_1(z),\ldots,+\hat{\kappa}_r(z)$. Since $\omega'(\hat{\kappa}_s(z))\neq0$, each singularity $\pm\hat{\kappa}_s(z)$ is simple and the corresponding residue can be computed as such:
	\begin{equation}
	\lim_{\kappa\to\pm\hat{\kappa}_s(z)}(\kappa\mp\hat{\kappa}_s(z))\frac{\G(\kappa)}{\omega(\kappa)-z}
	\end{equation}
	and
	\begin{eqnarray}
	\omega(\kappa)-z\sim\omega'(\pm\hat{\kappa}_s(z))(\kappa\mp\hat{\kappa}_s(z))\qquad(\kappa\to\pm\hat{\kappa}_s(z)),
	\end{eqnarray}
	thus
	\begin{equation}
	\lim_{\kappa\to\pm\hat{\kappa}_s(z)}(\kappa\mp\hat{\kappa}_s(z))\frac{\G(\kappa)}{\omega(\kappa)-z}=\frac{\G(\pm\hat{\kappa}_s(z))}{\omega'(\pm\hat{\kappa}_s(z))}\equiv \Z(\pm\hat{\kappa}_s(z)),
	\end{equation}
	where, because of the symmetry properties of $\omega(\kappa)$ and $\G(\kappa)$,
	\begin{equation}
	\Z(-\overline{\kappa})=-\Z(\kappa)^\dag,\qquad \Z(\overline{\kappa})=\Z(\kappa)^\dag.
	\end{equation}
	Applying the residue theorem, we finally obtain
	\begin{equation}
	\int_{-R}^R\frac{\G(k)}{\omega(k)-z}\,\mathrm{d}k=\int_{\delta C^+_R}\frac{\G(\kappa)}{\omega(\kappa)-z}\,\mathrm{d}\kappa+\int_{\Gamma^+_R}\frac{\G(\kappa)}{\omega(\kappa)-z}\,\mathrm{d}\kappa+2\pi\i\,\sum_{s=1}^r\Z(\hat{\kappa}_s(z)),
	\label{eq:selfr}
	\end{equation}	
	and, taking the limit $R\to\infty$ in~\eqref{eq:selfr}, the left-hand term converges to $\S(z)$ and the two $R$-dependent components in the right-hand term converge respectively to the null matrix and to $\mathsf{b}(z)$ because of our hypotheses. The desired result in~\eqref{eq:selfinf} is obtained by using $\S(\overline{z})=\S(z)^\dag$, 
	
	$(ii)$ For $s=1,\dots,r$, either $\Im\hat{\kappa}_s(E)>0$ (hence $s\in I^+(E)$) or $\Im\hat{\kappa}_s(E)=0$ (hence $s\in I^0(E)$). Now, the two sets of poles behave differently when we move continuously from $\A^+$ to $\A^-$ (see Fig.~\ref{fig:poles}):
	\begin{itemize}
		\item the poles $+\hat{\kappa}_s(E+\i\delta)$ with $s\in I^+(E)$ do not cross the real line when $\delta$ becomes negative, and hence they are still included in the contour, yielding a contribution to the self-energy which is continuous from $\A^+$ to $\A^-$;
		\item the poles $+\hat{\kappa}_s(E+\i\delta)$ with $s\in I^0(E)$ \textit{do} cross the real line when $\delta$ becomes negative, and hence they are not be included in the contour; the poles $-\hat{\kappa}_s(E+\i\delta)$ take their place, yielding a contribution to the self-energy which is discontinuous from $\A^+$ to $\A^-$.
	\end{itemize}
	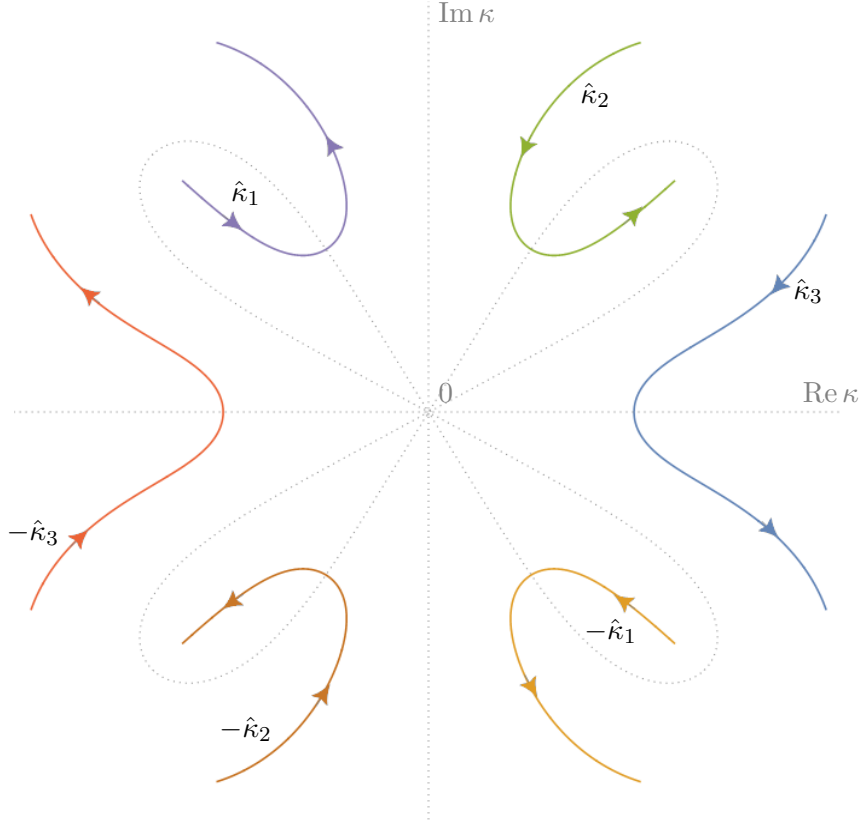
\begin{figure}[t]\centering
		\begin{tikzpicture}
		\node at (0,0) {\includegraphics[scale=0.45]{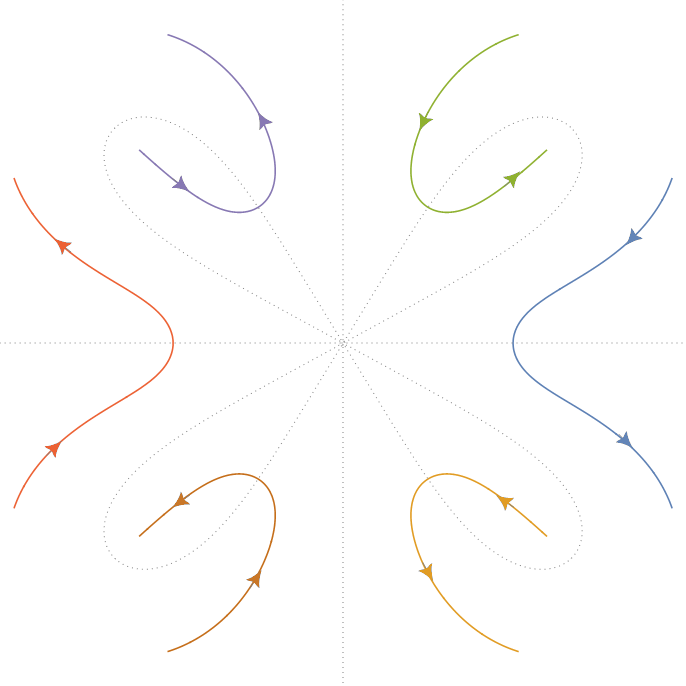}};
		\node[black] at (-2.4,-4.2) {$-\hat{\kappa}_2$};
		\node[black] at (2.4,-2.9) {$-\hat{\kappa}_1$};
		\node[black] at (-2.4,2.9) {$\hat{\kappa}_1$};
		\node[black] at (2.2,4.2) {$\hat{\kappa}_2$};
		\node[black] at (-5.2,-1.6) {$-\hat{\kappa}_3$};
		\node[black] at (5.0,1.6) {$\hat{\kappa}_3$};
		\node[above,gray] at (5.3,0) {$\Re\kappa$};
		\node[right,gray] at (0,5.3) {$\Im\kappa$};
		\node[above right,gray] at (0,0) {$0$};
		\end{tikzpicture}
		\caption{An example: solutions $\pm\hat{\kappa}_s(E_0+\i\delta)$, $s=1,2,3$, of the equation $\omega(\kappa)=E_0+\i\delta$ with $\omega(\kappa)=\frac{k^6}{k^4+1}$, with $E_0=0.02$ and $\delta\in[-1.4,1.4]$; the arrows point in the direction of decreasing $\delta$. The dotted gray lines corresponds to the set of points satisfying $\Im\omega(\kappa)=0$; each solution intersects one of those curves at $\delta=0$. The poles $\pm\kappa_1(E_0+\i\delta)$, $\pm\kappa_2(E_0+\i\delta)$ always belong to $\mathbb{C}^+$, whereas $\pm\kappa_3(E_0+\i\delta)$ is in $\mathbb{C}^\pm$ for $\delta>0$ and in $\mathbb{C}^\mp$ for $\delta<0$; consequently, $I^+(E_0)=\{1,2\}$ and $I^0(E_0)=\{3\}$.}
		\label{fig:poles}
	\end{figure}\FloatBarrier
	Recalling that $\mathsf{b}(\overline{z})=\mathsf{b}(z)^\dag$ and $\Z(-\kappa)=\Z(\overline{\kappa})^\dag$, we can hence write
	\begin{equation}
	\S(E+\i\delta)=\mathsf{b}(E+\i\delta)+2\pi\i\left(\sum_{s\in I^+(E)}\Z(\hat{\kappa}_s(E+\i\delta))+\sum_{s\in I^0(E)}\Z(\hat{\kappa}_s(E+\i\delta))\right);
	\end{equation}
	\begin{equation}
	\S(E-\i\delta)=\mathsf{b}(E-\i\delta)+2\pi\i\left(\sum_{s\in I^+(E)}\Z(\hat{\kappa}_s(E-\i\delta))+\sum_{s\in I^0(E)}\Z\Bigl(\overline{\hat{\kappa}_s(E-\i\delta)}\Bigr)^\dag\right),
	\end{equation}
	and, taking the limit $\delta\downarrow0$,
	\begin{equation}
	\S(E+\i0)=\mathsf{b}(E)+2\pi\i\left(\sum_{s\in I^+(E)}\Z(\hat{\kappa}_s(E))+\sum_{s\in I^0(E)}\Z(\hat{\kappa}_s(E))\right);
	\end{equation}
	\begin{equation}
	\S(E-\i0)=\mathsf{b}(E)+2\pi\i\left(\sum_{s\in I^+(E)}\Z(\hat{\kappa}_s(E))+\sum_{s\in I^0(E)}\Z(\hat{\kappa}_s(E))^\dag\right),
	\end{equation}
	the discontinuity being thus entirely due to the poles which cross the real line:
	\begin{equation}
	\S(E+\i0)-\S(E-\i0)=\i\sum_{s\in I^0(E)}\left(\Z(\hat{\kappa}_s(E))-\Z(\hat{\kappa}_s(E))^\dag\right),
	\end{equation}
	Besides, by Eq.~\eqref{eq:quartet}, for every $s\in I^+(E)$, there is $s'\in I^+(E)$ such that 
	\begin{equation}
	\hat{\kappa}_{s'}(E)=-\overline{\hat{\kappa}_s(E)}\implies \Z(\hat{\kappa}_{s'}(E))=\Z(-\overline{\hat{\kappa}_s(E)})=-\Z(\hat{\kappa}_s(E))^\dag,
	\end{equation}
	hence, because of the presence of the imaginary factor $2\pi\i$, the contribution of the poles $\hat{\kappa}_s(E)$ for $s\in I^+(E)$ is purely hermitian. The same argument cannot be applied to the poles with index $s\in I^0(E)$, whose contribution will therefore have both a hermitian and anti-hermitian part. 
	
	We finally get
	\begin{equation}
	\Re\S(E\pm\i0)=\mathsf{b}(E)+\i\left(\sum_{s\in I^+(E)}\Z(\hat{\kappa}_s(E))+\sum_{s\in I^0(E)}\Re\Z(\hat{\kappa}_s(E))\right);
	\end{equation}
	\begin{equation}
	\Im\S(E\pm\i0)=\pm\i\sum_{s\in I^0(E)}\Im\Z(\hat{\kappa}_s(E)).
	\end{equation}
	
	$(iii)$ is an immediate consequence of the continuity of $z\in\A\mapsto\hat{\kappa}_s(z)\in\K$.
\end{proof}

\section{Proof of Prop.~\ref{prop:prop2}}\label{app:proof2}
First of all, let us introduce the auxiliary matrix $\mathsf{\Lambda}_n(\phi_1,\phi_2,\dots,\phi_{n-1})$, with $\phi_1,\dots,\phi_{n-1}\in\C$, defined as
\begin{equation}
[\mathsf{\Lambda}_n(\phi_1,\dots,\phi_{n-1})]_{j\ell}=\begin{cases}1,&j=\ell;\\\prod_{k=\ell}^{j-1}\phi_{k},&j>\ell
\end{cases}
\end{equation}
and $[\mathsf{\Lambda}_n(\phi_1,\dots,\phi_{n-1})]_{\ell j}=[\mathsf{\Lambda}_n(\phi_1,\dots,\phi_{n-1})]_{j\ell}$, i.e., explicitly,
\begin{equation}
\mathsf{\Lambda}_{n}(\phi_1,\dots)=\begin{pmatrix}  
1&\phi_1&\phi_1\phi_2&\phi_1\phi_2\phi_3&\dots&\prod_{j=1}^{n-2}\phi_j&\prod_{j=1}^{n-1}\phi_j\\
\phi_1&1&\phi_2&\phi_2\phi_3&\dots&\prod_{j=2}^{n-2}\phi_j&\prod_{j=2}^{n-1}\phi_j\\
\phi_1\phi_2&\phi_2&1&\phi_3&\dots&\prod_{j=3}^{n-2}\phi_j&\prod_{j=3}^{n-1}\phi_j\\
\phi_1\phi_2\phi_3&\phi_2\phi_3&\phi_3&1&\dots&\prod_{j=4}^{n-2}\phi_j&\prod_{j=4}^{n-1}\phi_j\\
\vdots&\vdots&\vdots&\vdots&\ddots&\vdots&\vdots\\
\prod_{j=1}^{n-2}\phi_j&\prod_{j=2}^{n-2}\phi_j&\prod_{j=3}^{n-2}\phi_j&\prod_{j=4}^{n-2}\phi_j&\dots&1&\phi_{n-1}\\
\prod_{j=1}^{n-1}\phi_j&\prod_{j=2}^{n-1}\phi_j&\prod_{j=3}^{n-1}\phi_j&\prod_{j=4}^{n-1}\phi_j&\dots&\phi_{n-1}&1
\end{pmatrix}
\label{an}
\end{equation}
We have
\begin{equation}
\mathsf{\Phi}_n(\kappa)=\mathsf{\Lambda}_n\left(\e^{\i\kappa(x_2-x_1)},\dots,\e^{\i\kappa(x_n-x_{n-1})}\right),
\end{equation}
i.e.\ the matrix $\mathsf{\Phi}_n(\kappa)$ (as well as its characteristic polynomial) is simply obtained by setting $\phi_j=\e^{\i\kappa|x_{j+1}-x_j|}$ in the expression of the matrix $\mathsf{\Lambda}_n(\phi_1,\dots,\phi_{n-1})$.

We can now proceed with the proof.
\begin{proof}[Proof of Prop.~\ref{prop:prop2}]
	By Eq.~\eqref{an}, the last and penultimate rows of $\mathsf{\Lambda}_n(\phi_1,\dots,\phi_{n-1})$, as well as its last and penultimate column, are proportional through the factor $\phi_{n-1}$, except the term with $j=\ell=n$. Applying elementary properties of the determinant, we have
	\begin{equation}
	\begin{vmatrix}
	1-\lambda&\phi_1&\phi_1\phi_2&\phi_1\phi_2\phi_3&\dots&\prod_{j=1}^{n-2}\phi_j&\prod_{j=1}^{n-1}\phi_j\\
	\phi_1&1-\lambda&\phi_2&\phi_2\phi_3&\dots&\prod_{j=2}^{n-2}\phi_j&\prod_{j=2}^{n-1}\phi_j\\
	\phi_1\phi_2&\phi_2&1-\lambda&\phi_3&\dots&\prod_{j=3}^{n-2}\phi_j&\prod_{j=3}^{n-1}\phi_j\\
	\phi_1\phi_2\phi_3&\phi_2\phi_3&\phi_3&1-\lambda&\dots&\prod_{j=4}^{n-2}\phi_j&\prod_{j=4}^{n-1}\phi_j\\
	\vdots&\vdots&\vdots&\vdots&\ddots&\vdots&\vdots\\
	\prod_{j=1}^{n-2}\phi_j&\prod_{j=2}^{n-2}\phi_j&\prod_{j=3}^{n-2}\phi_j&\prod_{j=4}^{n-2}\phi_j&\dots&1-\lambda&\phi_{n-1}\\
	\prod_{j=1}^{n-1}\phi_j&\prod_{j=2}^{n-1}\phi_j&\prod_{j=3}^{n-1}\phi_j&\prod_{j=4}^{n-1}\phi_j&\dots&\phi_{n-1}&1-\lambda
	\end{vmatrix}=\nonumber
	\label{pn}
	\end{equation}
	\begin{equation}
	=\phi_{n-1}^2\begin{vmatrix}
	1-\lambda&\phi_1&\phi_1\phi_2&\phi_1\phi_2\phi_3&\dots&\prod_{j=1}^{n-2}\phi_j&\prod_{j=1}^{n-2}\phi_j\\
	\phi_1&1-\lambda&\phi_2&\phi_2\phi_3&\dots&\prod_{j=2}^{n-2}\phi_j&\prod_{j=2}^{n-2}\phi_j\\
	\phi_1\phi_2&\phi_2&1-\lambda&\phi_3&\dots&\prod_{j=3}^{n-2}\phi_j&\prod_{j=3}^{n-2}\phi_j\\
	\phi_1\phi_2\phi_3&\phi_2\phi_3&\phi_3&1-\lambda&\dots&\prod_{j=4}^{n-2}\phi_j&\prod_{j=4}^{n-2}\phi_j\\
	\vdots&\vdots&\vdots&\vdots&\ddots&\vdots&\vdots\\
	\prod_{j=1}^{n-2}\phi_j&\prod_{j=2}^{n-2}\phi_j&\prod_{j=3}^{n-2}\phi_j&\prod_{j=4}^{n-2}\phi_j&\dots&1-\lambda&1\\
	\prod_{j=1}^{n-2}\phi_j&\prod_{j=2}^{n-2}\phi_j&\prod_{j=3}^{n-2}\phi_j&\prod_{j=4}^{n-2}\phi_j&\dots&1&\frac{1-\lambda}{\phi_{n-1}^2}
	\end{vmatrix}=\nonumber
	\label{pn2}
	\end{equation}
	\begin{equation}
	=\phi_{n-1}^2\begin{vmatrix}
	1-\lambda&\phi_1&\phi_1\phi_2&\phi_1\phi_2\phi_3&\dots&\prod_{j=1}^{n-2}\phi_j&0\\
	\phi_1&1-\lambda&\phi_2&\phi_2\phi_3&\dots&\prod_{j=2}^{n-2}\phi_j&0\\
	\phi_1\phi_2&\phi_2&1-\lambda&\phi_3&\dots&\prod_{j=3}^{n-2}\phi_j&0\\
	\phi_1\phi_2\phi_3&\phi_2\phi_3&\phi_3&1-\lambda&\dots&\prod_{j=4}^{n-2}\phi_j&0\\
	\vdots&\vdots&\vdots&\vdots&\ddots&\vdots&\vdots\\
	\prod_{j=1}^{n-2}\phi_j&\prod_{j=2}^{n-2}\phi_j&\prod_{j=3}^{n-2}\phi_j&\prod_{j=4}^{n-2}\phi_j&\dots&1-\lambda&\lambda\\
	\prod_{j=1}^{n-2}\phi_j&\prod_{j=2}^{n-2}\phi_j&\prod_{j=3}^{n-2}\phi_j&\prod_{j=4}^{n-2}\phi_j&\dots&1&\frac{1-\lambda}{\phi_{n-1}^2}-1
	\end{vmatrix}=\nonumber
	\label{pn3}
	\end{equation}
	\begin{equation}
	=\phi_{n-1}^2\begin{vmatrix}
	1-\lambda&\phi_1&\phi_1\phi_2&\phi_1\phi_2\phi_3&\dots&\prod_{j=1}^{n-2}\phi_j&0\\
	\phi_1&1-\lambda&\phi_2&\phi_2\phi_3&\dots&\prod_{j=2}^{n-2}\phi_j&0\\
	\phi_1\phi_2&\phi_2&1-\lambda&\phi_3&\dots&\prod_{j=3}^{n-2}\phi_j&0\\
	\phi_1\phi_2\phi_3&\phi_2\phi_3&\phi_3&1-\lambda&\dots&\prod_{j=4}^{n-2}\phi_j&0\\
	\vdots&\vdots&\vdots&\vdots&\ddots&\vdots&\vdots\\
	\prod_{j=1}^{n-2}\phi_j&\prod_{j=2}^{n-2}\phi_j&\prod_{j=3}^{n-2}\phi_j&\prod_{j=4}^{n-2}\phi_j&\dots&1-\lambda&\lambda\\
	0&0&0&0&\dots&\lambda&\frac{1-\lambda}{\phi_{n-1}^2}-(1+\lambda)
	\end{vmatrix}
	\label{pn4}\end{equation}
	where we have divided the last row and column by $\phi_n$, subtracted the penultimate column from the last column, and finally subtracted the penultimate row from the last row. 
	
	By expanding the determinant along the last row (or column) we will get a term proportional to the determinant of $\mathsf{\Lambda}_{n-1}(\phi_1,\dots,\phi_{n-2})$, plus another term proportional to the determinant of a matrix obtained from $\mathsf{\Lambda}_{n-1}(\phi_1,\dots,\phi_{n-2})$ by replacing the last column (or row) with a column in which the only nonzero term is the one with $j=\ell=n-1$; a further expansion of the determinant yields a term proportional to the determinant of $\mathsf{\Lambda}_{n-2}(\phi_1,\dots,\phi_{n-3})$ and, substituting $\phi_n=\e^{\i\kappa(x_n-x_{n-1})}$, the recurrence formula~\eqref{recurr}.
	
	In particular, taking $\lambda=0$, Eq.~\eqref{recurr} simplifies as follows:
	\begin{equation}
	p_n(0,\kappa)=[1-\e^{2\i\kappa(x_n-x_{n-1})}]p_{n-1}(0,\kappa),
	\end{equation}
	which is a recurrence equation of the first order with constant coefficients, whose solution can be thus obtained:
	\begin{equation}
	p_n(0,\kappa)=\prod_{j=1}^{n-1}[1-\e^{2\i\kappa(x_{j+1}-x_{j})}]p_1(0,\kappa)=\prod_{j=1}^{n-1}[1-\e^{2\i\kappa(x_{j+1}-x_{j})}].
	\end{equation}
\end{proof}


\begin{thebibliography}{99}

	\bibitem{Araki}
	H. Araki, Y. Munakata, M. Kawaguchi, and T. Got$\hat{\text{o}}$, Quantum Field Theory of Unstable Particles. \textit{Progr. Theor. Phys.}  \textbf{17}(3), 419--442 (1957).
	
	\bibitem{Frohlich}
	V. Bach, J. Fr\"ohlich, and I. M. Sigal, Return to equilibrium. \textit{J. Math. Phys.} \textbf{41}, 3985 (2000).
	
	\bibitem{Bohm}
	A. Bohm,\textit{ Rigged Hilbert space and quantum mechanics}. Tech. rep. Texas Univ., Austin (USA). Center for Particle Theory, 1974.
	
	\bibitem{BohmGadella}
	A. Bohm, M. Gadella, and J. D. Dollard, \textit{Dirac kets, Gamow vectors and Gelfand triplets: the Rigged Hilbert space formulation of quantum mechanics}. Springer, 1989.
	
	\bibitem{breuer}
	H.-P. Breuer and F. Petruccione, \textit{The Theory of Open Quantum Systems}. Oxford University Press on Demand, Oxford, 2002.
	
	\bibitem{cicc} G. Calaj\`o, Y.-L. L. Fang, H. U. Baranger, and F. Ciccarello, Exciting a bound state in the continuum through multiphoton scattering plus delayed quantum feedback. \textit{Phys. Rev. Lett.} \textbf{122}, 073601 (2019).
	
	\bibitem{cohen} C. Cohen-Tannoudji, J. Dupont-Roc, and G. Grynberg, \textit{Atom-Photon Interactions}. New York: Wiley, 2010.
	
	\bibitem{Madrid}
	R. de la Madrid, The role of the rigged Hilbert space in quantum mechanics. \textit{Eur. J. Phys.} 26.2 (2005), pp. 287–312.
	
	\bibitem{MadridBohm}
	R. de la Madrid, A. Bohm, and M. Gadella, Rigged Hilbert space treatment of continuous spectrum. \textit{Fortschr. Phys.} 50.2 (2002), pp. 185–216.
	
	\bibitem{FL} P. Facchi, M. Ligab\`o and D. Lonigro, Spectral properties of the singular Friedrichs-Lee Hamiltonian. \textit{J. Math. Phys.} \textbf{62}, 032102 (2021).
	
	\bibitem{bic1} P. Facchi, M. S. Kim, S. Pascazio, F. V. Pepe, D. Pomarico, and T. Tufarelli, Bound states and entanglement generation in waveguide quantum electrodynamic. \textit{Phys. Rev. A} \textbf{94}, 043839 (2016).
	
	\bibitem{bic2} P. Facchi, D. Lonigro, S. Pascazio, F. V. Pepe, and D. Pomarico, Bound states in the continuum for an array of quantum atoms. \textit{Phys. Rev. A} \textbf{100}, 023834 (2019).
	
	\bibitem{Friedrichs}
	K. O. Friedrichs, On the perturbation of continuous spectra. \textit{Commun. Pur. Appl. Math.} \textbf{1}(4), 361--406 (1948).
	
	\bibitem{GadellaGomez}
	M. Gadella and F. G\'omez, On the mathematical basis of the Dirac formulation of Quantum Mechanics. \textit{Int. J. Theor. Phys.} 42.10 (2003), pp. 2225-2254.
	
	\bibitem{Gadella}
	M. Gadella and G. Pronko, The Friedrichs model and its use in resonance phenomena. \textit{Fortschr. Phys.} \textbf{59}(9), 795--859 (2011).
	
	\bibitem{Gesztesy2}
	F. Gesztesy and E. Tsekanovskii, On Matrix-Valued Herglotz Functions. \textit{Math. Nachr.} \textbf{218}(1), 61--138 (2000)
	
	\bibitem{Giacosa} F. Giacosa, The Lee model: a tool to study decays. arXiv:2001.07781 [hep-ph] (2020).
	
	\bibitem{Giacosa2} F. Giacosa, Non-exponential Decay in Quantum Field Theory and in Quantum Mechanics: The Case of Two (or More) Decay Channels. \textit{Found. Phys.} \textbf{42}, 1262–1299 (2012).
	
	\bibitem{Horwitz}
	L. P. Horwitz, Lee-Friedrichs model. \textit{Encyclopedia of Mathematics}, http://www.encyclopediaofmath.org/index.php?title=Lee-Friedrichs model\&oldid=22719 (1998).
	
	\bibitem{ingold}
	G.-L. Ingold, Path integrals and their application to dissipative quantum systems. In: \textit{Coherent Evolution in Noisy Environments},  Buchleitner, A., and Hornberger, K, editors, 1–53 (Springer, Berlin, Heidelberg, 2002).
	
	\bibitem{antizeno}
	A. G. Kofman, G. Kurizki, Frequent observations accelerate decay: The anti-Zeno effect. \textit{Z. Naturforsch.} A \textbf{56}, 83-90 (2001).
	
	\bibitem{zeno}
	K. Koshino, A. Shimizu, Quantum Zeno effect by general measurements. \textit{Phys. Rept.} \textbf{412}, 191-275 (2005).
	
	\bibitem{Lee}
	T. D. Lee, Some Special Examples in Renormalizable Field Theory. \textit{Phys. Rev.} \textbf{95}(5), 1329--1334 (1954).
	
	\bibitem{leggett}
	A. J. Leggett, S. Chakravarty, A. T. Dorsey, M. P. A. Fisher, A. Garg, and W. Zwerger, Dynamics of the dissipative two-state system. \textit{Reviews of Modern Physics} \textbf{59}(1), 1 (1987).

	\bibitem{liu}
	Z. Liu, W. Kamleh, D. B. Leinweber, F. M. Stokes, A. W. Thomas, and J. Wu, Hamiltonian Effective Field Theory Study of the $N$*(1535) Resonance in Lattice QCD. \textit{Phys. Rev. Lett}. \textbf{116}, 082004 (2016).
	
	\bibitem{palma} F. Lombardo, F. Ciccarello, and G. M. Palma, Photon localization versus population trapping in a coupled-cavity array. \textit{Phys. Rev. A} \textbf{89}, 053826 (2014).
	
	\bibitem{FLProc}
	D. Lonigro, P. Facchi and M. Ligab\`o, The Friedrichs-Lee Model and Its Singular Coupling Limit. \textit{Proceedings} \textbf{12}(1), 17 (2019).
	
	\bibitem{bic3} D. Lonigro, P. Facchi, S. Pascazio, F. V. Pepe, and D. Pomarico, Stationary excitation waves and multimerization in arrays of quantum emitters. \textit{New J. Phys.} \textbf{23}, 103033 (2021).
	
	\bibitem{Mathias}
	R. Mathias, Matrices with Positive Definite Hermitian Part: Inequalities and Linear Systems. \textit{SIAM J. Matrix Anal. Appl.} 13 (1992): 640-654.
	
	\bibitem{Jaynes}
	B. Shore and P. Knight, The Jaynes-Cummings Model. \textit{Journal of Modern Optics} \textbf{40}(7), 1195-1238 (1993).
	
	\bibitem{vonWaldenfels}
	W. von Waldenfels, \textit{A Measure Theoretical Approach to Quantum Stochastic Processes}. Vol. 878. Lecture Notes in Physics. Springer-Verlag Berlin Heidelberg, 2014.
				
	\bibitem{zhou}
	Z. Zhou, and Z. Xiao, Understanding $X$(3862), $X$(3872), and $X$(3930) in a Friedrichs-model-like scheme. \textit{Phys. Rev. D} \textbf{96}, 099905 (2017).
	
	\end{thebibliography}
\end{document}